\documentclass[preprint,3p]{elsarticle}
\usepackage{lineno,hyperref}
\usepackage{amsfonts} 
\usepackage{amsmath}
\usepackage{subcaption}
\usepackage{bm}
\usepackage{url}
\usepackage{longtable}
\usepackage{booktabs}

\journal{}
\usepackage{multicol}
\usepackage{color}

\usepackage{amsthm}
\newtheorem{definition}{Definition}
\newtheorem{assumption}{Assumption}
\newtheorem{theorem}{Theorem}
\newtheorem{problem}{Problem}

\newtheorem{proposition}{Proposition}

\biboptions{authoryear}
\newif \ifmargincomments 
\margincommentstrue

\ifmargincomments

\else

\fi
 \ifmargincomments

\else 

\fi
\begin{document}
\begin{frontmatter}
\title{Planning and Operations of Mixed Fleets in Mobility-on-Demand Systems}
\author[stanford]{Kaidi Yang}
\ead{kaidi.yang@stanford.edu}
\author[stanford]{Matthew W. Tsao}
\ead{mwtsao@stanford.edu}
\author[grab]{Xin Xu}
\ead{xin.xu@grabtaxi.com}
\author[stanford]{Marco Pavone}
\ead{pavone@stanford.edu}
\address[stanford]{Autonomous Systems Laboratory, Stanford University}
\address[grab]{Grab Holdings Inc. }
\begin{abstract}
Automated vehicles (AVs) are expected to be beneficial for Mobility-on-Demand (MoD), thanks to their ability of being globally coordinated. 
To facilitate the steady transition towards full autonomy, we consider the transition period of AV deployment, whereby an MoD system operates a mixed fleet of automated vehicles (AVs) and human-driven vehicles (HVs). 
In such systems, AVs are centrally coordinated by the operator, and the HVs might strategically respond to the coordination of AVs. 
We devise computationally tractable strategies to coordinate mixed fleets in MoD systems. Specifically, we model an MoD system with a mixed fleet using a Stackelberg framework where the MoD operator serves as the leader and human-driven vehicles serve as the followers. We develop two models: 1) a steady-state model to analyze the properties of the problem and determine the planning variables (e.g., compensations, prices, and the fleet size of AVs), and 2) a time-varying model to design a real-time coordination algorithm for AVs. The proposed models are validated using a case study inspired by real operational data of a MoD service in Singapore. Results show that the proposed algorithms can significantly improve system performance. 

\end{abstract}
\end{frontmatter}

\section{Introduction}


The past decade has witnessed the widespread deployment of Mobility-on-Demand services, thanks to the rapid adoption of smartphones, developments in wireless communication, and the boom of shared economies. These services, e.g., the ride-hailing services provided by Uber and DiDi,  present immense potential to enhance mobility and accessibility while reducing resource usage. One key operational challenge associated with these services is represented by the \emph{vehicle imbalances} due to asymmetric transportation demand: vehicles tend to accumulate in some regions while becoming depleted in  others, giving rise to inefficient operations of the MoD system.  Currently, MoD systems typically address this challenge by combining dynamic pricing \citep[see for example,][]{ZhaYinEtAl2018,BattifaranoQian2019,YangShaoEtAl2020, ChenZhengEtAl2020,MaXuEtAl2020} with a real-time heat-map of the passenger demand to rebalance their fleets. However, the rebalancing actions are still performed in a decentralized manner by human drivers who are interested in their own earnings, which may not yield optimal system performance. 

The emergence of automated vehicles (AVs) provides opportunities for sophisticated and centralized vehicle control, and thus can be beneficial to MoD systems. By integrating AVs into MoD systems, Autonomous Mobility-on-Demand (AMoD) is expected to be a promising paradigm for future mobility systems, whereby a fleet of autonomous taxi-like vehicles is coordinated by a central operator to provide on-demand mobility services.  Compared to traditional MoD services, AMoD offers several advantages. 
First, by eliminating the driver costs, AMoD can reduce the price of 
trips and thus improve the operator's profit. 
Second, since AVs can always be operational in the system, AMoD can provide continuous and reliable services regardless of the time of the day.
Third, unlike human drivers who can be ill-informed or self-interested, AVs can be coordinated in a centralized manner to provide better services which enable higher vehicle utilization and operational efficiency. 
Thanks to these advantages, AMoD has attracted increasing attention in the research fields of transportation and robotics, including demand analysis and prediction \citep{ChenLowEtAl2015,WenNassirEtAl2019}, real-time coordination algorithms~\citep{ZhangPavone2016,LiuBansalEtAl2018,TsaoIglesiasEtAl2018,IglesiasRossiEtAl2017}, interactions with public transport \citep{SalazarRossiEtAl2018} and power networks~\citep{TuckerTuranEtAl2019,EstandiaSchifferEtAl2019,BoewingSchifferEtAl2020}, transportation network design~\citep{PintoHylandEtAl2019,ZardiniLanzettiEtAl2020}, and system evaluation \citep{DiaJavanshour2017,Hoerl2018,FagnantKockelman2018,HorlRuchEtAl2019}.

Despite the benefits of AMoD systems, it is evident that AVs will only gradually be technologically mature and adopted in the market. 
During the transition period, MoD systems will conceivably be operating a mixed fleet of AVs and human-driven vehicles (HVs), whereby HVs might respond to the coordination of AVs strategically, making global optimization  challenging.  
To facilitate the steady transition towards full autonomy, this paper aims to devise computationally tractable strategies to design and coordinate mixed fleets in MoD systems, considering the interactions between AVs and HVs. Specifically, 
we frame a fleet coordination problem with a mixed equilibrium of AVs that are centrally coordinated and HVs that act according to their own interests.  From an operational perspective, such a system could also be interpreted as a mixed system of compliant drivers (e.g., contractor drivers who are paid to strictly follow the instructions given by the operator) and self-interested drivers--thus, the  tools and insights derived in this paper could be applied to existing systems (i.e., without AVs) as well.

\emph{Related work}. To the best of our knowledge, 
such a mixed fleet system has been rarely studied in the context of MoD services \citep{LokhandwalaCai2018,WeiPedarsaniEtAl2019}. 
\citet{LokhandwalaCai2018} analyzed the ride-sharing serviced provided by autonomous taxis and human-driven taxis based on an agent-based simulation of New York City, considering factors such as taxi shifts and passenger preference. 
\citet{AfecheLiuEtAl2018} focused on a two-location, four route loss network and investigated the impact of demand-side admission and supply-side rebalancing control on the spatial vehicle imbalances and the strategical behavior of drivers. 
The steady-state system equilibria are established in scenarios with different control regimes, ranging from minimal platform control to centralized admission and repositioning control. 
\citet{WeiPedarsaniEtAl2019} explicitly considered the interactions between AVs and HVs and analyzed the steady-state behavior of the mixed fleet system in a transportation network with equi-distant nodes, whereby the behavior of HVs is modeled after 
\citet{BimpikisCandoganEtAl2019} as a set of non-linear equations. 
Both \citet{AfecheLiuEtAl2018} and \citet{WeiPedarsaniEtAl2019} focused on the steady-state analysis of special types of transportation networks. To sum up, it remains unclear how mixed fleet systems with realistic road networks can be controlled, especially in real time. 

Several works analyze the behavior of HVs in the MoD context. Most works focus on the route choice of HVs delivering passengers  without considering the rebalancing behavior or the willingness of HVs to accept passengers \citep[e.g.,][]{ZhuLevinson2015,HeChenEtAl2017,LiWangEtAl2018}, or 
analyze the response of HVs to incentives such as surge or dynamic pricing at a single or at a few locations \citep{BanerjeeJohariEtAl2015, CastilloKnoepfle2017, Chen2016, ZhaChenEtAl2017,GudaSubramanian2019}. 
\citet{Buchholz2019}, on the other hand, empirically analyzed the dynamic spatial equilibrium of taxicabs based on the New York City taxi data, and showed that simple change in prices can improve the taxi services. 
\citet{BimpikisCandoganEtAl2019} studied the rebalancing behavior of HVs by formulating the equilibrium as the solution to a set of non-linear equations. However, it is assumed that the human drivers always accept the passengers assigned to them by the operator, which may not be the case in  real-world systems. Moreover, it is challenging to leverage the model proposed in \citet{BimpikisCandoganEtAl2019} for real-time control due to its non-linear structure.

 Related works on such mixed fleet systems also exist in the context of traffic assignment (i.e., without dispatch or rebalancing), which typically search for the optimal equilibrium in a Stackelberg game where the leader is the group of compliant drivers, and the followers are the self-interested drivers~\citep{YangZhangEtAl2007,FlorianMorosan2014,XieXie2014,ChenHeEtAl2017,BagloeeSarviEtAl2017}. These works, however, focus on congestion games between the two types of vehicles. It is nevertheless unclear how the proposed algorithms in these works can be adapted to the MoD services where rebalancing and  passenger assignment are important features.

\emph{Statement of contribution}. The contribution of this paper is three-fold. First, we initiate  research on mixed fleet systems with realistic road networks in a MoD context. We account for the interactions between AVs and HVs using a Stackelberg framework where the MoD operator with its AV fleet serves as the leader and human-driven vehicles serve as the followers. Second, we develop two models: 1) a steady-state model to analyze the properties of the problem and determine the planning variables (e.g., prices, compensations for HVs, and the fleet size of the AVs), and 2) a time-varying model to design a real-time coordination algorithm. Third, we conduct real-world case studies using real data in Singapore to validate the proposed algorithms and provide a guideline to deploy AVs in scenarios with various AV penetration rates. 

This paper is organized as follows. Section \ref{sec:system} presents an overview of MoD systems with a mixed fleet and introduces general notions. Section \ref{sec:ss} provides theoretical analysis in the steady state, establishes the equilibrium of the mixed fleet, and derives decisions for planning. Section \ref{sec:tv} builds upon the results of Section \ref{sec:ss} and develops a simplified model for real-time control. Section \ref{sec:simulation} presents simulation results to illustrate the benefits of employing AVs in the system. Section \ref{sec:conclusion} concludes the paper and proposes future directions.

\section{System Description}
\label{sec:system}
For reader's convenience, we summarize the most important variables in the following table. 

\begin{longtable}{p{0.8in}p{5.2in}}
\caption{List of most important variables}\\
\multicolumn{2}{c}{The following variables appear in both the steady-state and time-varying formulations } \\\toprule
$\mathcal{G}=(\mathcal{N},\mathcal{E})$ & graph representation of the transportation network, where the node set $\mathcal{N}$ is the set of stations and the edge set $\mathcal{E}$ is the set of shortest paths connecting the stations. \\\hline
$A$ & adjacency matrix of graph $\mathcal{G}$.\\\hline
$\mathcal{T}$ & set of discrete time intervals $\mathcal{T}=\{1,2,\cdots, T\}$. \\\hline
$\Delta T$ & length of time step. \\\hline
$\mathcal{M}$ & set of vehicle classes $\mathcal{M}=\{a,h\}$, indexed by $m$, where $a$ represents automated vehicles (AVs) and $h$ represents human-driven vehicles (HVs). \\\hline
$\tau_{ij}$ & travel time for vehicles traveling from station $i$ to station $j$.\\\hline
$\delta_{ij}$ & travel distance for vehicles traveling from station $i$ to station $j$.\\\hline
$N^m$ & upper bound on the number of vehicles of vehicle class $m\in \mathcal{M}$ in the transportation system. \\\hline
$\sigma$ & cost of operating vehicles of any class per vehicle per distance driving. \\\hline
$\mu$ & minimum earning rate such that HVs are willing to enter the system, referred as the value of time (VOT) of HV drivers. \\\hline
$q_{ijt}$ & passenger demand between origin-destination (OD) pair $(i,j) \in \mathcal{E}$ starting at time step $t\in\mathcal{T}$. \\\hline
$c_{ijt}$ & compensation for HVs for serving the requests between OD pair $(i,j) \in \mathcal{E}$  starting at time step $t\in\mathcal{T}$ \\\hline
$p_{ijt}$ & prices for the requests between OD pair $(i,j) \in \mathcal{E}$  starting at time step $t\in\mathcal{T}$. \\\hline
$x_{ijt}^m$& passenger flow, i.e., number of passengers with OD pair $(i,j)\in\mathcal{E}$ served by vehicle class $m\in \mathcal{M}$ within time step $t\in\mathcal{T}$. \\\hline
$y_{ijt}^m$& rebalancing flow, i.e., number of vehicles of class $m\in\mathcal{M}$ rebalancing  from station $i$ to station $j$ within time step $t\in\mathcal{T}$. \\\hline
$v$& earning rate of HVs at equilibrium. \\\hline
$\phi_i$& earnings of HVs at equilibrium starting from station $i$.\\\hline
$\Phi$ & function indicating HV behavior. \\ 
\bottomrule\\

\multicolumn{2}{c}{The following variables appear only in the steady-state formulation}\\\toprule
$\bar{q}_{ij} $ & remaining demand for HVs after the operator allocates demand to AVs. \\\hline
$F_{ij}(\cdot)$ & cumulative distribution function of the willingness to pay for OD pair $(i,j)\in\mathcal{E}$. \\\hline
$H_{ij}(\cdot)$ & inverse demand function for OD pair $(i,j)\in\mathcal{E}$. \\\hline
$u_{ij}$ & number of HVs waiting to serve trips from station $i$ to station $j$. \\\hline
$\zeta_{ij}$ & waiting time for OD pair $(i,j)\in\mathcal{E}$. \\\hline
$\mathcal{L}$ & set of trails of HVs (see Definition~\ref{def:trail}), indexed by $l$.\\\hline
$U_l$ & total earning along trail $l \in\mathcal{L}$.\\\hline
$T_l$ & total traversing time along trail $l\in\mathcal{L}$. \\\hline
\bottomrule\\

\multicolumn{2}{c}{The following variables appear only in the time-varying formulation}\\\toprule
$w_{ijt}$& number of passengers with OD pair $(i,j)\in\mathcal{E}$ waiting to be served at time step $t\in\mathcal{T}$. \\\hline
$r_{it}$& number of vacant vehicles at station $i\in\mathcal{N}$ at time step $t$. \\\hline
$\psi$& VOT of passengers. \\\hline 
$\epsilon$ & per time step probability that unserved passengers choose to stay in the system, $\epsilon\in [0,1]$. \\\hline
$\lambda$ & weight parameter that combines the objective of the leader model with the objective of the follower model. \\\bottomrule
\end{longtable}

\begin{figure}[htbp]
    \centering
    \includegraphics[width=0.8\textwidth]{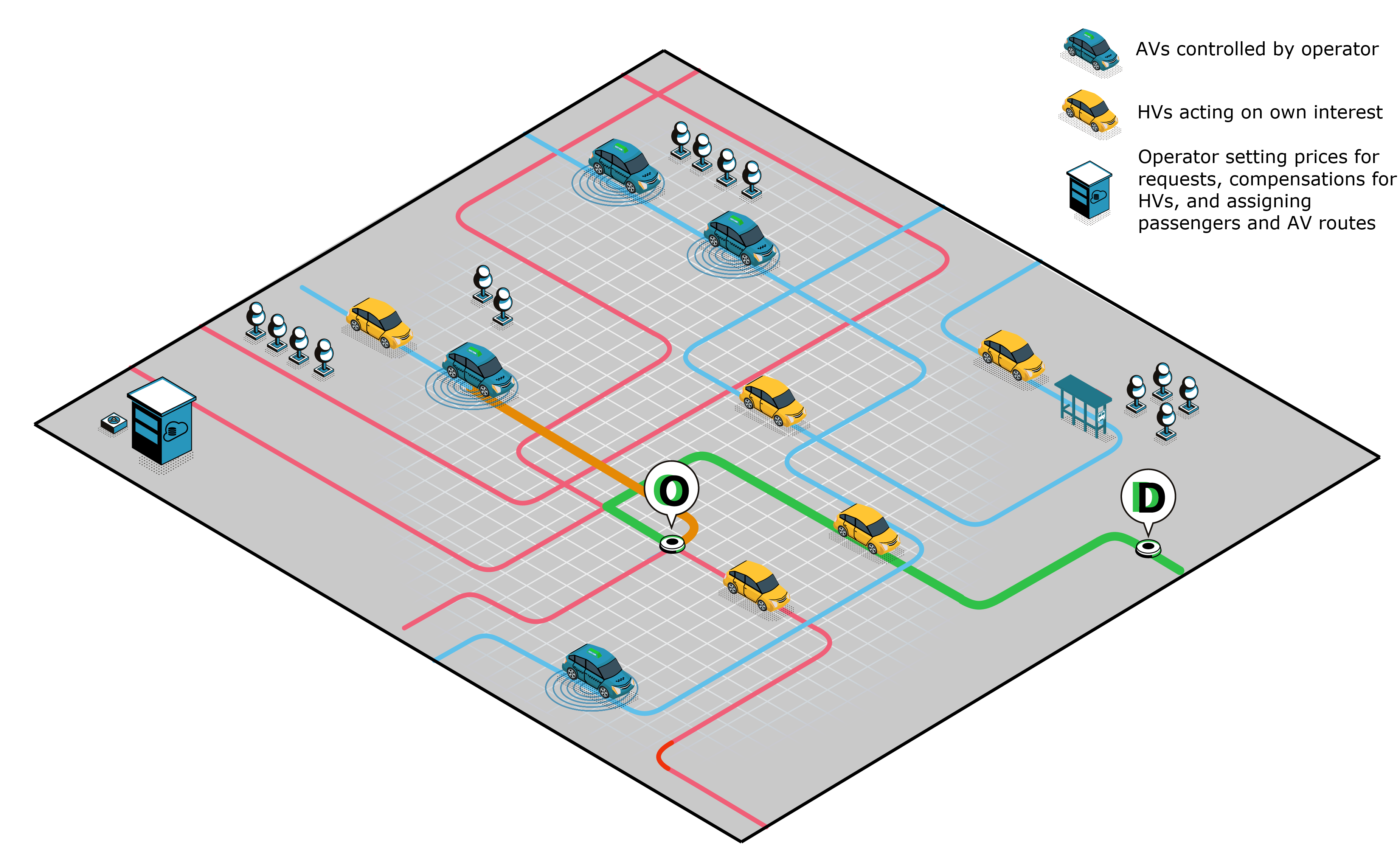}
    \caption{Illustration of the MoD system, where blue vehicles represent automated vehicles (AVs) and yellow vehicles represent human-driven vehicles (HVs). An operator is able to set prices for request and compensations for HVs, assign passengers, and coordinate the routes of AVs.  }
    \label{fig:system}
\end{figure}

We consider a city where an operator provides mobility on-demand (MoD) services with a mixed fleet of automated vehicles (AVs, denoted as $a$) and human-driven vehicles (HVs, denoted as $h$), illustrated in Figure~\ref{fig:system}. 
Mathematically,  we describe the urban transportation network as a weighted graph $\mathcal{G}=(\mathcal{N},\mathcal{E})$, where $\mathcal{N}$ is the set of stations (i.e., pick-up or drop-off locations) and $\mathcal{E}$ is a set of the directed edges, i.e., shortest paths between pairs of stations. For a typically city, we consider $\mathcal{G}$ to be fully connected such that a directed edge exists for each pair of stations. Denote matrix $A\in\mathbb{R}^{|N|\times|N|}$ as the adjacency matrix
of graph $\mathcal{G}$, where $|\cdot|$ represents the number of elements in a set.  We discretize the time horizon into a set of discrete intervals $\mathcal{T}=\{1,2,\cdots, T\}$ of a given length $\Delta T$. 
Without loss of generality, we assume both types of vehicles can operate on each station of the network. The proposed methodological framework can be readily adapted to account for the limited driving capacity of AVs during the transition period, whereby AVs can only operate within certain regions or take certain paths.  
 

The travel time for edge $(i,j)\in\mathcal{E}$ is defined as the number of time steps it takes a vehicle to travel along the shortest path between station $i$ and station $j$, denoted as an integer $\tau_{ij} \in \mathbb{Z}_+$. We make the following remarks for the travel times. First, we assume that the travel times are given and independent of the coordination of the MoD fleet. This assumption applies to cities where the MoD fleet constitutes a relatively small proportion of the entire vehicle population on the transportation network, and thus the impact of the MoD fleet on traffic   is marginal. Second, we observe that the travel times typically change slowly with time, and hence for simplicity of notation, we ignore the time dimension of $\tau_{ij}$. The proposed modeling framework can be easily adapted to consider time-varying exogenous travel times. Third, we assume the travel times to be independent of the vehicle classes, since vehicles of both classes have similar speeds and will follow the same shortest paths between two stations. The modeling framework can also be easily adapted to incorporate class-dependent travel times. Similar to travel times, we also define the travel distance for edge $(i,j)\in\mathcal{E}$ as the length of the shortest path between station $i$ and station $j$, denoted as $\delta_{ij}$.

We allow the number of each vehicle class to be bounded by an upper bound $N^m,~m\in\mathcal{M}$, where $N^a$ is the number of AVs owned by the operator, and $N^h$ represents the restriction for HVs. Unlike the existing works that assume infinite HV provision~\citep{WeiPedarsaniEtAl2019,BimpikisCandoganEtAl2019},  we consider the scenarios where certain restrictions can be imposed on the number of HVs in the transportation system. One of such scenarios is that traffic authorities limit the number of vehicles on the streets to alleviate traffic congestion, reduce emissions, or improve the earnings of drivers, e.g., in  New York City~\citep{NYC2019}. 

The AV fleet is always available for service operation, whereas HVs determine their own availability. 
HVs would only participate in the system if their expected net earnings (the difference between their earnings and operational costs) per time step is higher than a given value $\mu$, which can be seen as the combination of drivers' value of leisure and drivers' external earnings if they do not participate in the system. Herein, we refer to this value $\mu$ as the value of time (VOT) of HV drivers. 
Furthermore, the operation of vehicles with any class incurs a fixed operational cost (e.g., the energy cost, vehicle depreciation, and vehicle maintenance) $\sigma$  per vehicle per unit distance driving. 

Passengers make transportation requests at each time step. We denote the origin-destination (OD) pairs as tuples $(i,j),~i,j \in\mathcal{N}$, and the demand for OD pair $(i,j)$ starting at time step $t\in\mathcal{T}$  as $q_{ijt}\geq 0$. Clearly, passengers that depart at time step $t$ arrive at their destinations at time step $t+\tau_{ij}$. We denote $w_{ijt}$ as the number of waiting passengers at time step $t$ that have not been matched with any vehicle. For  simplicity of presentation, we consider the assumption as in the existing literature~\citep{WeiPedarsaniEtAl2019} that passengers can only be taken by vehicles located at the same station as their origins. Nevertheless, we can 
straightforwardly adapt the proposed models to relax this assumption (see Section~\ref{sec:tvpickup}).

With the passenger requests, the MoD operator determines its strategies to optimize system performance (e.g., its profits, system-level earnings, social welfare, etc.). 
From the planning perspective, the operator can set the prices $p_{ijt}\geq 0$ for the requests with OD pair $(i,j)$ served at time $t$, the compensation $c_{ijt}\geq 0$ to HVs that serve these requests, and the fleet size of AVs in operation. 
From the operational perspective, the operator can fully coordinate the AVs, dynamically assigning passengers to them, and determining their routes.  
HVs will then freely choose the remaining passengers to optimize their own interests. This setting applies to systems where the operator broadcasts the passenger requests to drivers (e.g., in many taxi systems) or systems where the operator makes recommendations, but HVs can reject the assigned requests with minimal penalties. The models proposed in this paper, nevertheless, can be tailored to analyze the scenarios where HVs always accept the assigned requests, or be extended to consider scenarios with heterogeneous HV fleets where different types of HVs (compliant or strategic in passenger assignment) co-exist. 
After a vehicle (AV or HV) is matched with a request, the vehicle will pick up the matched passengers, and then deliver them to their destinations. 
Vehicles not matched with any passengers may either stay at the same station or rebalance to other stations. The rebelancing decisions for AVs are made by the operator, whereas HVs make such decisions on their own.
Mathematically, let us define the passenger flow $x_{ijt}^m\geq 0 $ and rebalancing  flow $y_{ijt}^m\geq 0$ as the number of vehicles of class $m\in\mathcal{M}$ that start moving along edge $(i,j)$ within time step $t$ with and without passengers, respectively. 
Passengers that are not matched can either stay in the system and enter the next round of assignment or leave. Let $\epsilon\in [0,1]$ be the per time step probability that unserved passengers choose to stay in the system, which characterizes the impatience of passengers. We assume that the probability $\epsilon$ is determined exogenously and does not depend on the real-time operation of the service (or equivalently, the coordination algorithms). Let  $w_{ijt}$ be the number of passengers waiting to be assigned at the beginning of time step $t$. Notice that $(1-\epsilon) w_{ijt}$ passengers will abandon their requests within time step $t$.

As a summary, the operator aims to optimize system performance by 1) determining the prices for passenger requests, compensations for HVs, and the fleet size of AVs; and 2) coordinating the passenger assignment and routing of AVs. HVs, on the other hand, make three types of decisions to maximize their earnings: 1) the participation in  the system, 2) the passenger pick-up strategy, and 3) the rebalancing strategy. 

Following the sequential property of such systems, we model the operations of mixed fleet systems using a Stackelberg game framework, where the leader is the MoD operator, and the followers are the HVs. We will present the details of the Stackelberg game in the rest of the paper, for both the steady-state formulation used for planning and the time-varying formulation used for real-time control.

\section{Planning of  MoD Systems With Mixed Fleets: Steady-State Formulation}
\label{sec:ss}
In this section, we propose a steady-state formulation of MoD systems with mixed fleets, based on the notion of Stackelberg games. In this formulation, we consider static demand and supply, and therefore ignore the time indices in all variables\footnote{More specifically, the demand, compensations, and prices for requests with OD pair $(i,j)$  become $q_{ij}$, $c_{ij}$, and $p_{ij}$, respectively. The passenger flow and rebalancing flow for OD pair $(i,j)$ and vehicle class $m$ become $x_{ij}^m$ and $y_{ij}^m$, respectively. The variables at steady state can be seen as the average of the real-time values.}.
We formulate the leader and follower models, provide a theoretical analysis of the equilibrium of the mixed fleet, and determine the planning variables, i.e., the trip prices, compensations, and the fleet size of AVs. 
The purpose of the steady-state formulation is two-fold. First, it helps the operator with making long-term plans on the number of AVs deployed  in the system. 
Second, the theoretical results of this section set the foundation for the real-time control model in Section~\ref{sec:tv}. 
This section is organized as follows. 
Section~\ref{sec:ssHV} develops the follower model and establishes the results on the equilibrium for HVs. 
Section~\ref{sec:ssOperator} presents the leader model, i.e., the optimization model solved by the operators. 
Section~\ref{sec:ssMixed} analyzes the equilibrium of the mixed fleet.

\subsection{Equilibrium of HVs}
\label{sec:ssHV}
In this subsection, we develop the follower model, which characterizes the behavior of HVs. 
We establish the equilibrium given remaining demand $0\leq \bm{\bar{q}} =\{\bar{q}_{ij}\}_{(i,j)\in\mathcal{E}}$ and compensations $0\leq \bm{c}=\{c_{ij}\}_{(i,j)\in\mathcal{E}}$. 
Without loss of generality, we only consider the case where there is at least one OD pair such that $\bar{q}_{ij}>0$, otherwise no HVs would appear in the system and the equilibrium is trivial. 
Then with the passenger flow $\bm{x}^h=\{x_{ij}^h\}_{(i,j)\in\mathcal{E}}$ and the rebalancing flow $\bm{y}^h=\{y_{ij}^h\}_{(i,j)\in\mathcal{E}}$, we have the following conservation equation for the passenger demand that HVs take:
\begin{align}
x_{ij}^h \leq \bar{q}_{ij},~(i,j)\in\mathcal{E}, \label{eq:HVDemand}
\end{align}
where the inequality represents the possibility that not all passengers can be served because 1) the number of HVs on the road network may be upper bounded; and 2) HVs would choose the requests to maximize their own interest, and may not be willing to serve the requests with low compensations. 

HVs that do not take any passenger would either rebalance themselves to other stations or stay at the same station. The passenger flow $\bm{x}^h$ and the rebalancing flow $\bm{y}^h$ satisfy the following conservation equation Eq.(\ref{eq:HVFlow}):
\begin{align}
\sum_{j \in \mathcal{N}\backslash\{i\}}x_{ji}^h  + \sum_{j \in \mathcal{N}\backslash\{i\} }y_{ji}^h  =  \sum_{j \in \mathcal{N} \backslash\{i\}}x_{ij}^h   + \sum_{j \in \mathcal{N}\backslash\{i\}}y_{ij}^h , ~i \in \mathcal{N},  \label{eq:HVFlow}
\end{align}
where Eq.(\ref{eq:HVFlow}) represents that, in steady state, the number of HVs entering a station equals the number of HVs leaving the station within each time step. 
 
HVs that decide to stay at station $i$ are interested in serving passenger requests with certain destinations. If passengers with some OD pairs are more valuable than others, HVs might choose to stay at the origin of these requests and wait to serve them. Let us denote $\bm{u}=\{u_{ij}\}_{(i,j)\in\mathcal{E}}$ with $u_{ij}\geq 0,~i,j \in \mathcal{N}$, as the number of HVs at station $i$ that are waiting to serve requests with an OD pair $(i,j)$. Clearly, a HV would wait for requests with OD pair $(i,j)$ (i.e., $u_{ij}>0$)  only if the demand $\bar{q}_{ij}$ has been fully served by other HVs. Such a relation can be represented by the following complementary condition, 
\begin{align}
u_{ij} (x_{ij}^h - \bar{q}_{ij}) = 0,~(i,j)\in\mathcal{E}. \label{eq:HVWaiting}
\end{align}
HVs waiting for serving requests with OD pair $(i,j)$ can be regarded as a queue, with a discharging rate equal to the passenger demand $\bar{q}_{ij}$. Hence, the waiting time for OD pair $(i,j)$ can be written as $\zeta_{ij} = u_{ij}/\bar{q}_{ij}$. 

The constraints on the total number of vehicles in the system can be described as
\begin{align}
\sum_{(i,j)\in\mathcal{E}} (x_{ij}^h+y_{ij}^h)\tau_{ij} + \sum_{(i,j)\in\mathcal{E}}u_{ij} \leq N^h,  \label{eq:HVUB}
\end{align}
where $N^h$ is the imposed upper bound on the number of HVs. The first term represents the number of HVs in passenger or rebalancing routes, and the second term represents the number of  HVs waiting for passengers.

It appears reasonable to assume that HVs enter and leave the system from the same station (e.g., their residence), so that their trail in the system is cyclic. We define the trail of HVs as in Definition \ref{def:trail}. 
 
\begin{definition}[Trail of HVs]\label{def:trail}
Denote $e=(i,j,\kappa)$ with $\kappa\in\{\rm{pax},~\rm{reb}\}$ as the trip HVs make between station $i$ and station $j$. We refer to $e$ as a passenger trip if 
$\kappa=\rm{pax}$ and as a rebalancing trip if $\kappa=\rm{reb}$. 
A trail of HVs $l=(e_1,e_2,\cdots,e_{|l|})\in\mathcal{L}$ is defined as a sequence of trips made by HVs such that $i_{r+1}=j_{r},~r=0,\cdots,|l|-1$, and $i_0=j_{|l|}$, where $|l|$ represents the number of trips in trail $l$. 
We say a trail is used if $ z_l =  \min_{r}\{x_{i_rj_r}\mathbb{I}_{\kappa_r=\rm{pax}} + y_{i_rj_r}\mathbb{I}_{\kappa_r=\rm{reb}}\}> 0$ and unused if $z_l=0$. 
\end{definition}

Following Definition~\ref{def:trail}, we can calculate the total earnings $U_l$ and the total traversing time $T_l$ along trail $l$ via Eq.(\ref{eq:U}) and Eq.(\ref{eq:T}), respectively,
\begin{align}
U_l &= \sum_{r=1}^{|l|}\Big((c_{i_rj_r}-\sigma\delta_{i_rj_r})\mathbb{I}_{\kappa_r=\rm{pax}}-\sigma\delta_{i_rj_r}\mathbb{I}_{\kappa_r=\rm{reb}}\Big),~l\in\mathcal{L}, \label{eq:U}\\
T_l &= \sum_{r=1}^{|l|}\Big((\tau_{i_rj_r}+\zeta_{i_rj_r})\mathbb{I}_{\kappa_r=\rm{pax}} + \tau_{i_rj_r}\mathbb{I}_{\kappa_r=\rm{reb}} \Big),~l\in\mathcal{L}. \label{eq:T}
\end{align}
Specifically, Eq.(\ref{eq:U}) calculates $U_l$ as the expected earning HVs make by taking trail $l$, defined as the difference between compensations and operational costs. The first term of Eq.(\ref{eq:T}) represents the time it takes to serve passenger trip $e_r$, including the travel time $\tau_{i_rj_r}$ and waiting time $\zeta_{i_rj_r}$, and the second term of Eq.(\ref{eq:T}) represents the time spent on rebalancing trips. 

We consider a scenario where HVs are interested in maximizing their \emph{earning rate}, defined as the ratio between the total earnings and the total travel time as in Definition \ref{def:earning}. 
  \begin{definition}[Earning rate]\label{def:earning}
  The earning rate of trail $l$, $V_l$, is defined as the ratio between the total earning $U_l$ and the total travel time $T_l$, i.e. $V_l=U_l/T_l$.  
 \end{definition}

We can now formally define the equilibrium of HVs as in Definition \ref{def:HVEquilibrium}. 

\begin{definition}[Equilibrium of HVs]\label{def:HVEquilibrium}
An equilibrium of HVs given compensations $\bm{c}$ and remaining demand $\bm{\bar{q}}$ is a tuple $(\bm{x}^h,\bm{y}^h,\bm{u})$ with $\bm{x}^h,\bm{y}^h,\bm{u}\geq 0$ satisfying the following conditions. 
\begin{enumerate}
\item[(a)~] Eq.(\ref{eq:HVDemand}) -- Eq.(\ref{eq:HVUB}) hold. 
\item[(b)~] For any used trail $l$, it holds that $U_{l'}/T_{l'}\leq U_l/T_l,~\forall l'\in\mathcal{L}$ and $\mu\leq U_l/T_l$, where $U_l$ and $T_l$ are defined in Eq.(\ref{eq:U}) and Eq.(\ref{eq:T}), respectively. 
\item[(c)~] For any used trail $l$, either $\mu = U_l/T_l$ or  Eq.(\ref{eq:HVUB}) holds with an equality. 
\end{enumerate}
\end{definition}

In Definition~\ref{def:HVEquilibrium}, condition (b) is consistent with the Waldrop's principle that the earning rate of any used trail is no less than the earning rate of any unused trail, i.e., only the best trails are chosen. Condition (b) also ensures that the earning rate of a used trail is no less than the VOT of HV drivers $\mu$. This condition implies that every HV in the system would have the same earning rate $v\geq \mu$. Condition (c) states that HVs will keep participating in the system until the earning rate equals the operational cost or the number of vehicles reaches the upper bound.   Denote by function $\Phi(\bm{c},\bm{\hat{q}})$ the set of equilibria defined in Definition~\ref{def:HVEquilibrium}.

Before we establish the equilibrium of HVs, we make the following assumption on the  
total earning of trails $U_l,l\in\mathcal{L}$, i.e., 
\begin{assumption}[Non-trivial equilibrium]\label{asm:trail}
There exists a trail $l\in\mathcal{L}$, such that $U_l\geq \mu T_l$. 
\end{assumption}
Assumption~\ref{asm:trail} states that at least one trail $l\in\mathcal{L}$  yields an earning rate no less than the VOT of HVs $\mu$. This ensures that some HVs would be interested in participating in the system. Otherwise, the equilibrium is trivial, whereby no HVs enter the system.  

It is challenging to compute an equilibrium in the sense of Definition~\ref{def:HVEquilibrium}, since there are an infinite number of trails, and Eq.(\ref{eq:HVWaiting}) represents a non-convex constraint. To tractably compute system equilibria, we present two auxiliary problems, namely  Problem~\ref{prb:HV1} and Problem~\ref{prb:HV2} (defined below). These two problems are linear programming (LP) problems characterized by the expected earning rate $v$. We will show that if the value of $v$ is properly chosen (e.g., via a line search algorithm), one can use the optimal primal and dual solutions to Problem~\ref{prb:HV1} and Problem~\ref{prb:HV2} to construct an equilibrium of HVs. These two problems are used to account for different cases of expected earning rate $v$. 


\begin{problem}[Auxiliary follower problem with given earning rate] \label{prb:HV1}
Given a fully connected transportation network $G=(\mathcal{N},\mathcal{E})$ with travel times $\bm{\tau}$, compensation $\bm{c}$, remaining demand $\bm{\bar{q}}$, costs associated with operating vehicles  per unit distance $\sigma$, and expected earning rate $v$, the auxiliary follower problem reads as follows: 
\begin{subequations}
\begin{align}
\max_{\bm{x}^h,\bm{y}^h} &\quad J^{\rm{SS}}_{L,v} = \sum_{(i,j)\in \mathcal{E}}c_{ij}x_{ij}^h - \sigma\sum_{(i,j)\in \mathcal{E}}\delta_{ij}(x_{ij}^h+y_{ij}^h) -v\sum_{(i,j)\in \mathcal{E}}\tau_{ij}(x_{ij}^h+y_{ij}^h) \label{eq:HVOpt0}\\
\rm{s.t.}&\quad x_{ij}^h \leq \bar{q}_{ij},~(i,j)\in\mathcal{E} \label{eq:HVOpt1}\\
&\quad \sum_{j \in \mathcal{N}\backslash\{i\}}x_{ji}^h  + \sum_{j \in \mathcal{N}\backslash\{i\} }y_{ji}^h  =  \sum_{j \in \mathcal{N} \backslash\{i\}}x_{ij}^h   + \sum_{j \in \mathcal{N}\backslash\{i\}}y_{ij}^h , ~i \in \mathcal{N}  \label{eq:HVOpt2} \\
&\quad x_{ij}^h,y_{ij}^h \geq 0,~(i,j)\in\mathcal{E}.  \label{eq:HVOpt3}
\end{align} \label{eq:HVOpt} 
\end{subequations}
\end{problem}

In Problem~\ref{prb:HV1}, the objective function Eq.(\ref{eq:HVOpt0}) represents the net earnings of HVs, where the first two terms represent the total earnings of HVs as the difference between the compensations (the first term) and operational costs (the second term), and the third term represents a proxy for expected earning of an outside option (i.e. an average trail). Since HVs are self interested, the third term encourages minimization of opportunity cost. Notice that the objective function does not consider the waiting time $\zeta_{ij}$ (or equivalently the number of waiting HVs $u_{ij}$). We will show that, at equilibrium, these variables can be represented using the dual solution to Problem~\ref{prb:HV1}.  Constraints Eq.(\ref{eq:HVOpt1}) and Eq.(\ref{eq:HVOpt2}) play the same role as Eq.(\ref{eq:HVDemand}) and Eq.(\ref{eq:HVFlow}), respectively, and constraints Eq.(\ref{eq:HVOpt3}) ensures that variables are nonnegative. Notice that, for a given expected earning rate $v$, Problem~\ref{prb:HV1} is a linear programming (LP) problem.  

\begin{problem}[Auxiliary follower problem with given earning rate and HV availability constraints] \label{prb:HV2}
Given a fully connected transportation network $G=(\mathcal{N},\mathcal{E})$ with travel times $\bm{\tau}$, compensation $\bm{c}$, remaining demand $\bm{\bar{q}}$, costs associated with operating vehicles  per unit distance $\sigma$, and expected earning rate $v$, the auxiliary follower problem considering the constraints on the number of HVs reads as follows: 
\begin{subequations}
\begin{align}
\max_{\bm{x}^h,\bm{y}^h} &\quad \tilde{J}^{\rm{SS}}_{L,v} = \rm{Eq.(\ref{eq:HVOpt0})}\notag \\
\rm{s.t.}&\quad \rm{Eq.(\ref{eq:HVOpt1})-\rm{Eq.(\ref{eq:HVOpt3})}}\notag \\
&\quad \frac{1}{v}\Big(\sum_{(i,j)\in \mathcal{E}}c_{ij}x_{ij}^h- \sigma\sum_{(i,j)\in \mathcal{E}}\delta_{ij}(x_{ij}^h+y_{ij}^h)\Big) = N^h \label{eq:HVOptNum}
\end{align}  
\end{subequations}
\end{problem}

Problem~\ref{prb:HV2} is also a LP. The only difference between Problem \ref{prb:HV1} and Problem \ref{prb:HV2} is the introduction of the constraints Eq.(\ref{eq:HVOptNum}). Notice that the earnings of HVs result only from the compensations. Hence, Eq.(\ref{eq:HVOptNum}) can be interpreted as the number of HVs in the system (see Theorem~\ref{thm:HVEquilibrium} for the formal proof of this equivalence). 

We aim to find an appropriate value of $v$ such that the optimal primal and dual solutions to Problem~\ref{prb:HV1} and Problem~\ref{prb:HV2} can be used to construct an equilibrium for the HVs. Notice that Problem~\ref{prb:HV1} and Problem~\ref{prb:HV2} do not include the number of waiting HVs $\bm{u}$, nor consider the constraint on the total number of vehicles in the system Eq.(\ref{eq:HVUB}) or the complementary condition Eq.(\ref{eq:HVWaiting}). We will determine $v$ such that the number of waiting HVs $\bm{u}$ can be derived using the dual variables of Problem~\ref{prb:HV1} and Problem~\ref{prb:HV2}, such that the constraint Eq.(\ref{eq:HVWaiting}) and Eq.(\ref{eq:HVUB}) hold.  Theorem~\ref{thm:HVEquilibrium} represents the main result of this section, 

\begin{theorem}[Equilibrium of HVs]\label{thm:HVEquilibrium}
Denote $(\bm{x}^{h*},\bm{y}^{h*})$ as the optimal solution to Problem~\ref{prb:HV1} with parameter $v>0$. Denote $\bm{\pi}^{*}=\{\pi^*_{ij}\}_{(i,j)\in\mathcal{E}}$ and $\bm{\phi}^{*}=\{\phi^*_{i}\}_{i\in\mathcal{N}}$ as the optimal dual variables corresponding to constraints Eq.(\ref{eq:HVOpt1}) and Eq.(\ref{eq:HVOpt2}), respectively. Set the waiting number of HVs $\bm{u}^*=\{u^*_{ij}\}_{(i,j)\in\mathcal{E}}$ as $u^*_{ij}=\pi^*_{ij}\bar{q}_{ij}/v,~(i,j)\in\mathcal{E}$. 
Define function $g(v)$ as
\begin{align}
g(v) = \sum_{(i,j)\in\mathcal{E}} (x^*_{ij}+y^*_{ij})\tau_{ij} +\sum_{(i,j)\in\mathcal{E}}u^*_{ij} . \label{eq:HVnumber}
\end{align}
Then the following statements hold: 
\begin{enumerate}
\item[(i)~] $g(v)$ is strictly decreasing if $v>0$, i.e., $g(\bar{v})< g(\hat{v})$, if $0<\bar{v}<\hat{v}$. The left-hand limit $g(v^-)$ and the right-hand limit $g(v^+)$ exists for every $v$, satisfying $g(v^+)\leq g(v^-)$,  and the discontinuities are countable. \item[(ii)~] Under Assumption~\ref{asm:trail}, there exists a $v_0>0$ such that $g(v_0^+)\leq N^h \leq g(v_0^-)$. 
\item[(iii)~] Consider a $v_0$ such that $g(v_0^+)\leq N^h \leq g(v_0^-)$, then $v = \max\{v_0,\mu\}$ is the earning rate at  equilibrium. If $v_0<\mu$ or $g(v^+) = g(v^-)$, $(\bm{x}^{h*},\bm{y}^{h*},\bm{u}^*)$ is an equilibrium of HVs. Otherwise, let 
$(\bm{\tilde{x}}^{h*},\bm{\tilde{y}}^{h*})$  be
the optimal solution to Problem~\ref{prb:HV2} with parameter $v$, then $(\bm{\tilde{x}}^{h*},\bm{\tilde{y}}^{h*})$ is an optimal solution to Problem~\ref{prb:HV1}, and there exists $\bm{\tilde{u}}^*$ such that  $(\bm{\tilde{x}}^{h*},\bm{\tilde{y}}^{h*},\bm{\tilde{u}}^*)$ is an equilibrium of HVs. 
\end{enumerate}
\end{theorem}
\begin{proof}
See \ref{proof:HVEquilibrium}. 
\end{proof}
We make the following remarks for Theorem \ref{thm:HVEquilibrium}. First, $g(v)$ characterizes the number of HVs in the system. Statement (i) demonstrates that $g(v)$ is a strictly decreasing function in $v$. This statement is intuitive, since
if more HVs participate in the system, each driver is expected to make less earnings due to the competition for requests.  Notice that $g(v)$ may not be 
continuous for certain values of $v$. This is because Problem~\ref{prb:HV1} may have multiple optimal solutions, and each solution may correspond to a specific number of HVs. Nevertheless, for any optimal solution to Problem~\ref{prb:HV1} with parameter $v$, the number of HVs are within the range $[g(v^+),g(v^-)]$. 
Second, statement (ii) exploits  monotonicity to find a proper $v_0$ such that the imposed restriction on the number of HVs can be satisfied with this $v_0$, i.e., $g(v_0^+)\leq N \leq g(v_0^-)$. Such a $v_0$ can be found efficiently using a bisection search algorithm.  
 Third, statement iii) shows that some solutions for Problem~\ref{prb:HV1} can be used to construct an equilibrium
 of HVs in the sense of Defintion~\ref{def:HVEquilibrium}.
 We consider two cases to account for the constraint of the imposed upper bound of HVs, i.e.,  Eq.(\ref{eq:HVUB}). In the first case, we show that $(\bm{x}^{h*},\bm{y}^{h*},\bm{u}^*)$ derived from the optimal primal and dual solutions to Problem~\ref{prb:HV1} is an equilibrium satisfying Eq.(\ref{eq:HVUB}). In the second case, there are multiple solutions to Problem~\ref{prb:HV1}, and we aim to find the solution $(\bm{\tilde{x}}^{h*},\bm{\tilde{y}}^{h*})$ such that Eq.(\ref{eq:HVUB})  holds with an equality, i.e., $g(v)=N^h$. To this end, we find such solutions as the  optimal solutions to Problem~\ref{prb:HV2} and use them to construct an equilibrium of HVs in the sense of Definition~\ref{def:HVEquilibrium}, where $\bm{\tilde{u}}^*$ can be computed with the dual variables associated with constraint Eq.(\ref{eq:HVOpt2}) at solution $(\bm{\tilde{x}}^{h*},\bm{\tilde{y}}^{h*})$.

\subsection{Operator's optimization model}\label{sec:ssOperator}
In this subsection, we present the leader model, i.e., the optimization problem solved by the MoD operator. 
The operator determines the prices for passengers $\bm{p} = \{p_{ij}\}_{(i,j)\in\mathcal{E}}$, compensations for HVs $\bm{c}=\{c_{ij}\}_{(i,j)\in\mathcal{E}}$, passenger flow for AVs $\bm{x}^a=\{x^a_{ij}\}_{(i,j)\in\mathcal{E}}$ and rebalancing flow for AVs $\bm{y}^a=\{y^a_{ij}\}_{(i,j)\in\mathcal{E}}$ to optimize system-level performance, while accounting for  the response of HVs.  

Passengers send requests with OD pair $(i,j)\in \mathcal{E}$ if their willingness to pay is at least the price $p_{ij}$. We characterize the induced demand with a given price $p_{ij}$ as $q_{ij} = \eta_{ij}(1-F_{ij}(p_{ij}))$, where $F_{ij}(\cdot)$ represents the cumulative distribution function (CDF) of the willingness to pay, for which we make the following assumption:
\begin{assumption}[CDF of the willingness to pay]\label{asm:CDF}
For each OD pair $(i,j)\in\mathcal{E}$, the CDF of the willingness to pay $F_{ij}$ is continuous and strictly increasing within its support. 
\end{assumption}

By Assumption~\ref{asm:CDF}, we can obtain an inverse demand function $p_{ij} = H_{ij}(q_{ij}),~q_{ij}\in [0, \eta_{ij}]$, where $ H_{ij}=F_{ij}^{-1}(1-q_{ij}/\eta_{ij})$ is a strictly decreasing function. 
We further make the following assumption for $ H_{ij}$, which holds true for many  commonly used CDF of the willingness to pay (e.g., uniform distribution, exponential distribution, log-normal distribution, etc.):
\begin{assumption}[Concave revenue function]\label{asm:inverseCDF}
For each OD pair $(i,j)\in \mathcal{E}$, the inverse demand function $H_{ij}(q),~q\in [0, \eta_{ij}]$ satisfies the condition that the revenue $H_{ij}(q)q$ is a concave function in $q$. 
\end{assumption}

The operator aims to optimize its utility, accounting for the reaction of HVs. Notice that although HVs can be self-interested, they generally have no inclination to undermine the utility of the operator. Hence, we assume that HVs are cooperative in response to the decision of the operator, as stated in Assumption~\ref{asm:HV}. 
\begin{assumption}[Cooperative HVs]
\label{asm:HV}
HVs are cooperative with the fleet operator. In other words, if there are multiple equilibria resulting in the same costs for the HVs, they choose the action  that maximizes the utility of the operator. 
\end{assumption}

We are now in a position to state the optimization problem for the fleet operator, where we consider the system-level earnings objective--this choice is made to model scenarios that account for not only  the profit of the operator, but also the welfare of drivers and the utility of passengers  (e.g., due to a direct participation of a public entity) and to align the theoretical results with the case study, with the understanding that the theory herein can be extended to also account for a profit objective (i.e., the operator-level earnings). 

\begin{problem}[Operator's optimization problem]\label{prb:operatorOpt}
Given a fully connected transportation network $G=(\mathcal{N},\mathcal{E})$ with travel times $\bm{\tau}=\{\tau_{ij}\}_{(i,j)\in\mathcal{E}}$, inverse demand functions $H_{ij},~(i,j)\in\mathcal{E}$, upper bounds for the number of vehicles $N^m,~m\in\mathcal{M}$, costs associated with operating vehicles  per unit distance $\sigma$, and the VOT of HV drivers $\mu$,  the operator determines its pricing and operational strategies for AVs, namely $(\bm{q},\bm{c},\bm{x}^a,\bm{y}^a)$, by solving the following optimization problem (\ref{eq:operatorOpt}): 
\begin{subequations}
\begin{align}
\max_{\bm{q},\bm{c},\bm{x}^a,\bm{y}^a,\bm{x}^h}&\quad J^{\rm{SS}}_{L} =
\sum_{(i,j)\in \mathcal{E}}H_{ij}(q_{ij})(x^h_{ij} +x^a_{ij})
-\mu\Big(\sum_{(i,j)\in\mathcal{E}} (x_{ij}^h+y_{ij}^h)\tau_{ij} + \sum_{(i,j)\in\mathcal{E}}u_{ij}\Big) \notag \\
& \quad \quad \quad \quad \quad \quad 
- 
\sigma\sum_{m\in\mathcal{M}}\sum_{(i,j)\in \mathcal{E}}\delta_{ij}(x^m_{ij}+y^m_{ij})
\label{eq:AVObj}\\
\rm{s.t.}&\quad x^a_{ij} \leq q_{ij}, ~(i,j)\in\mathcal{E} \label{eq:AVDemand} \\
&\quad \sum_{j\in\mathcal{N}\backslash\{j\}}x_{ij}^a + \sum_{j \in \mathcal{N}\backslash\{i\}}y_{ij}^a = \sum_{j\in\mathcal{N}\backslash\{j\}}x_{ji}^a + \sum_{j \in \mathcal{N}\backslash\{i\}}y_{ji}^a, ~i  \in \mathcal{N} \label{eq:AVFlow} \\
& \quad\sum_{(i,j)\in \mathcal{E}}\tau_{ij}(x^a_{ij}+y^a_{ij}) \leq N^a\label{eq:AVSize} \\
&\quad \bm{x}^h,\bm{y}^h,\bm{u} \in \Phi(\bm{c},\bm{q}-\bm{x}^a) \label{eq:HVCouple} \\
&\quad 0\leq q_{ij} \leq \eta_{ij}, ~(i,j)\in\mathcal{E} \label{eq:demandBounds} \\
&\quad x_{ij}^a\geq 0,~y^a_{ij}\geq 0, ~(i,j)\in\mathcal{E}.\label{eq:AVBounds}
\end{align} \label{eq:operatorOpt}
\end{subequations}
\end{problem}
In Problem \ref{prb:operatorOpt}, the objective function Eq.(\ref{eq:AVObj}) defines the system-level earnings, defined as the difference between the earnings from the passengers (the first term) and the operational costs, including the cost of time for the HVs (the second term) and the operational costs of both fleets (the third term). 
Constraints Eq.(\ref{eq:AVDemand}) ensure that the demand served by AVs does not exceed the total demand. 
Constraints Eq.(\ref{eq:AVFlow}) represent the conservation of AVs in steady state, meaning that the incoming AV flow to a station equals  the outgoing AV flow. The left hand side (LHS) represents the incoming flow to station $i$, and the right hand side (RHS) represents the outgoing flow from station $i$. Constraint Eq.(\ref{eq:AVSize}) sets an upper bound on the number of AVs deployed in the system. Constraints Eq.(\ref{eq:HVCouple}) calculate the passenger flow of HVs $\bm{x}^h=\{x_{ij}^h\}_{(i,j)\in\mathcal{E}}$, where function $\Phi$ represents the follower model that outputs the equilibria of HVs with respect to compensations $\bm{c}$ and remaining demand $\bm{q}-\bm{x}^a$ as per Section \ref{sec:ssHV}. 
Notice that by Assumption~\ref{asm:HV}, function $\Phi$ will output the equilibrium that favors the 
Notice that constraints Eq.(\ref{eq:HVCouple}) also ensure that $\bm{x}^h$ satisfies the conservation of passengers and HV flows.
Constraints Eq.(\ref{eq:demandBounds}) set a bound on the passenger demand. Constraints Eq.(\ref{eq:AVBounds}) ensure that all variables are nonnegative.

\subsection{Equilibrium of the Mixed Fleet}
\label{sec:ssMixed}

In this subsection, we characterize the equilibrium for the mixed fleet. 
Following the notion of Stackelberg games, we formally define the equilibrium for the MoD system  
with mixed fleets as follows, 
\begin{definition}[Equilibrium for the MoD system with mixed fleets]
\label{def:mixed}
An equilibrium for the MoD system with a mixed fleet of AVs and HVs is a tuple $(\bm{p}^*,\bm{c}^*,\bm{x}^{a*},\bm{y}^{a*},\bm{x}^{h*},\bm{y}^{h*},\bm{u}^*)$, satisfying:
\begin{enumerate}
\item[(a)~] tuple $(\bm{x}^{h*},\bm{y}^{h*},\bm{u}^*)$ is an equilibrium for HVs, and
\item[(b)~] $(\bm{q}^*,\bm{c}^*,\bm{x}^{a*},\bm{y}^{a*})$ is an optimal solution to the optimization problem for the fleet operator Problem~\ref{prb:operatorOpt}, where  $q^*_{ij}=\eta_{ij}(1-F_{ij}(p^*_{ij})),~(i,j)\in\mathcal{E}$.
\end{enumerate}
\end{definition}

For an equilibrium as defined in Definition \ref{def:mixed}, we observe that the number of waiting HVs $u_{ij}^{*}=0$, meaning that the compensations are chosen such that no HVs would be waiting for any OD  pair $(i,j)$, as otherwise the operator could always decrease the compensation $c_{ij}$ such that the number of waiting HVs $u_{ij}^{*}=0$ without affecting the number of passengers served for this OD pair (see Eq.(\ref{eq:HVWaiting})). The key challenge to compute an equilibrium as per Definition \ref{def:mixed}, is that Problem ~\ref{prb:operatorOpt} is a non-convex optimization problem, and thus difficult to solve. In the following we introduce a surrogate optimization problem,  namely Problem \ref{prb:global}, which is convex and allows one to tractably compute an equilibrium.
 
\begin{problem}[Surrogate optimization problem] \label{prb:global}
Given a fully connected transportation network $G=(\mathcal{N},\mathcal{E})$ with travel times $\bm{\tau}$, inverse demand functions $H_{ij},~(i,j)\in\mathcal{E}$, upper bounds on the number of vehicles $N^m,~m\in\mathcal{M}$, costs associated with operating vehicles  per unit distance $\sigma$, the VOT of HV drivers $\mu$, the surrogate problem entails optimizing  the  pricing and operational strategies of both AVs and HVs, namely $(\bm{q},\bm{x}^a,\bm{y}^a,\bm{x}^h,\bm{y}^h)$ by solving:
\begin{subequations}
\begin{align}
\max_{\bm{q},\bm{x}^a,\bm{y}^a,\bm{x}^h,\bm{y}^h}&\quad \tilde{J}^{\rm{SS}}_L = \sum_{(i,j)\in \mathcal{E}}q_{ij}H_{ij}(q_{ij}) - \sigma\sum_{(i,j)\in \mathcal{E}}\sum_{m\in\mathcal{M}}\delta_{ij}(x^m_{ij}+y^m_{ij})-\mu\sum_{(i,j)\in \mathcal{E}}\tau_{ij}(x^h_{ij}+y^h_{ij})\label{eq:AVObj2}\\
\rm{s.t.}&\quad x^a_{ij} + x^h_{ij} = q_{ij}, ~(i,j)\in\mathcal{E} \label{eq:AVDemand2} \\
&\quad \sum_{j\in\mathcal{N}\backslash\{j\}}x_{ij}^m + \sum_{j \in \mathcal{N}\backslash\{i\}}y_{ij}^m = \sum_{j\in\mathcal{N}\backslash\{j\}}x_{ji}^m + \sum_{j \in \mathcal{N}\backslash\{i\}}y_{ji}^m, ~i  \in \mathcal{N},~m\in\mathcal{M} \label{eq:AVFlow2} \\
& \quad\sum_{(i,j)\in \mathcal{E}}\tau_{ij}(x^m_{ij}+y^m_{ij}) \leq N^m,~m\in\mathcal{M}\label{eq:AVSize2} \\
&\quad 0\leq q_{ij} \leq \eta_{ij}, ~(i,j)\in\mathcal{E} \label{eq:demandBounds2} \\
&\quad x_{ij}^m\geq 0,~y^m_{ij}\geq 0, ~(i,j)\in\mathcal{E},~m\in\mathcal{M}\label{eq:AVBounds2}
\end{align} \label{eq:globalOpt}
\end{subequations}
\end{problem}
Note that under Assumption~\ref{asm:inverseCDF}, Problem \ref{prb:global} is a convex optimization problem. We next show that one can choose compensations $\bm{c}$, which along with the optimal solution to Problem \ref{prb:global}, provide an equilibrium in the sense of Definition \ref{def:mixed}.
\begin{theorem}[Computation of mixed-fleet equilibrium]\label{thm:mixed}
Suppose Assumption~\ref{asm:inverseCDF} holds. Denote $(\bm{q}^*,\bm{x}^{a*},\bm{y}^{a*},\bm{x}^{h*},\bm{y}^{h*})$ as the the optimal solution to Problem~\ref{prb:global}. Let $u^{*}_{ij}=0$ and $p^*_{ij}=H_{ij}(p_{ij}),~(i,j)\in\mathcal{E}$. Then, there exist compensations $\bm{c}^*$ such that $0\leq c^*_{ij}\leq p^*_{ij}$ and  $(\bm{p}^*,\bm{c}^*,\bm{x}^{a*},\bm{y}^{a*},\bm{x}^{h*},\bm{y}^{h*})$ is an equilibrium for the MoD system with mixed fleets. 
  
\end{theorem}
\begin{proof}
See \ref{proof:mixed}. 
\end{proof}
While the full proof of this theorem is provided in the Appendix, here we highlight that the compensations $\bm{c}^*$ can be readily obtained from the dual variables of constraint Eq.(\ref{eq:AVDemand2}). Notice that since Problem~\ref{prb:global} globally optimizes system-level earnings, one can infer from Theorem~\ref{thm:mixed} that the chosen compensations can incentivize HVs to behave in a system-optimal manner. 

To summarize, we have modeled a mixed fleet MoD systems operating AVs and HVs as a Stackelberg game, whereby the operator serves as the leader to  optimize the quality of service and the HVs serve as the self-interested followers to optimize their earnings. The equilibria of such a game are defined in Definition~\ref{def:mixed}, while Theorem~\ref{thm:mixed} provides a numerically efficient procedure to compute them.

\section{Real-time Control of MoD Systems With Mixed Fleets}
\label{sec:tv}

In this section, we present real-time control algorithms for MoD systems with mixed fleets. We assume that  prices $\bm{p}$ and compensations $\bm{c}$ have already been determined as the result of a planning phase  (either through the model introduced in Section \ref{sec:ss}, or by other means), and devise a Stackelberg game-based Model Predictive Control (MPC) approach to coordinate AVs in real-time, with the objective of optimizing system-level earnings (as defined in the previous section). Section~\ref{sec:tvMPC} presents the MPC formulation, and Section~\ref{sec:tvpickup} refines it to more explicitly consider the pickup process. 

\subsection{A Stackelberg game-based MPC approach} \label{sec:tvMPC}

We develop an MPC approach for the real-time control of  MoD systems with mixed fleets. The proposed MPC approach relies on an  embedded optimization model based on a Stackelberg game, where the leader is the MoD operator which optimizes  passenger and rebalancing routes for the AVs in order to improve system-level earnings, while the followers are the HVs which strategically respond to the operator's decisions. 
At each time step, we take as an input the predicted passenger demand and the vehicle states (i.e., the number of idle vehicles or the vehicles en-route for passenger pickup or delivery), and solve an optimization problem over a  receding time horizon to compute passenger and rebalancing routes for the AVs that maximize system-level earnings. As is typical for MPC-style algorithms, only the passenger and rebalancing routes of AVs at the current time step are executed, and the process is repeated. This mechanism has the advantage of taking future system performance into account when optimizing current actions. 

Unlike the steady-state formulation, we characterize the rebalancing strategy of HVs as  rebalancing probabilities $P_{ijt}$ for HVs to move from station $i\in\mathcal{N}$ to station $j\in\mathcal{N}$ at time $t\in\mathcal{T}$, where station $j$ can be the same or different from station $i$. We consider such probabilities to be determined \emph{externally} in the MPC formulation, i.e., they are predicted from experience (e.g., from historical data). 
The reasons for this modeling choice are as follows. 
First, we expect HVs not to have  real-time global information about passenger demand at other stations or positions of other vehicles, and thus they can only slowly adapt their rebalancing strategy after experiencing  difficulty in getting  passengers at some stations. 
As the MPC considers a relatively short time horizon (e.g., 20min), it appears reasonable to assume that HVs stick to the planned rebalancing strategies over such a horizon. 
Second, we expect the MPC algorithm, due to its repeated optimizations, to be robust to small errors in the prediction--robustness will be experimentally evaluated in Section~\ref{sec:simulation}. 
Third, from a computational standpoint, this choice  significantly reduces the number of decision variables, and thus makes the algorithm much more scalable. Essentially, with this modeling choice we simplify the representation of the HV response to the passenger flows of HVs, while still retaining its essential features.

Denote by $r_{it}^m$  the number of vacant vehicles of class $m\in\mathcal{M}$ located at station $i\in\mathcal{N}$ at time step $t$. Denote by $w_{ijt}$ the number of passengers with OD $(i,j)$ who are waiting to be matched with drivers at time step $t$. We establish the leader model and follower model in the MPC formulation as Problem~\ref{prb:rvleader} and Problem~\ref{prb:rvfollower}, respectively. 

\begin{problem}[Leader model in the MPC formulation]\label{prb:rvleader}
 Let $K$ be the length of the planning horizon, $\mathcal{T}_0=\{t_0,t_0+1,~t_0+K-1\}$ be the set of time steps in the planning horizon from time step $t_0$, and $\mathcal{T}_{-}$ be the set of time steps prior to time step $t_0$. Given a tuple $\{\bm{q}_t,\bm{p}_t,\bm{c}_t\}_{t\in \mathcal{T}_0}$ of predicted demand, compensations, and prices, as well as the initial conditions $(\bm{r}_0^a,\bm{r}_0^h,\{\bm{x}_t^h,\bm{y}_t^h,\bm{x}_t^a,\bm{y}_t^a\}_{t\in\mathcal{T}_{-}})$, the operator determines the passenger and rebalancing flow of AVs $\{\bm{x}_t^a,\bm{y}_t^a\}_{t\in\mathcal{T}_{0}}$ by solving the optimization problem:
\begin{subequations}
\begin{align}
\max_{\bm{x}^a, \bm{y}^a, \bm{w}, \bm{r}^a, \bm{r}^h} \quad & J_L^{\rm{TV}} = \sum_{m \in \mathcal{M}} \sum_{t\in\mathcal{T}_0}\sum_{(i,j) \in \mathcal{E}} p_{ijt}x_{ijt}^m- \sigma\sum_{m \in \mathcal{M}}\sum_{t\in\mathcal{T}_0}\sum_{(i,j) \in \mathcal{E}}\delta_{ij}x_{ijt}^m - \sigma\sum_{t \in\mathcal{T}_0}\sum_{(i,j)\in\mathcal{E}}{\delta_{ij}(y_{ijt}^a+P_{ijt}r_{i,t}^{h})} \notag \\
&\quad\quad\quad - \psi\sum_{t\in\mathcal{T}}\sum_{(i,j) \in \mathcal{E}}w_{ij}^t \label{eq:rvleader0}\\
\rm{s.t.} \quad  & r_{i,t+1}^{a}  = r_{it}^{a} + \sum_{j \in \mathcal{N}_i}y_{ji,t-\tau_{ji}}^{a} - \sum_{j \in \mathcal{N}_i}y_{ijt}^{a} + \sum_{(i,j) \in \mathcal{E}}x_{ji,t-\tau_{ji}}^{a}-  \sum_{(i,j) \in \mathcal{E}}x_{et}^{a},~i \in\mathcal{N},~t\in\mathcal{T}_0 \label{eq:rvleader1}\\
& r_{i,t+1}^{h} = \sum_{j\in\mathcal{N}}P_{jit}r_{j,t}^{h} + \sum_{(i,j) \in \mathcal{E}}x_{ji,t-\tau_{ji}}^{h}-  \sum_{(i,j) \in \mathcal{E} }x_{ijt}^{h},~i \in\mathcal{N},~t\in\mathcal{T}_0 \label{eq:rvleader2}  \\
& w_{ij,t+1} = \epsilon w_{ijt} + q_{ijt} - x_{ijt}^{a}-x_{ijt}^{h},~(i,j) \in \mathcal{E},~t\in\mathcal{T}_0 \label{eq:rvleader3} \\
& \bm{x}^h = \Phi\Big(\bm{w}+\bm{q}-\bm{x}^a,\bm{r}^h\Big) \label{eq:rvleader4} \\
&\bm{x}^a, \bm{y}^a, \bm{w}, \bm{r}^a, \bm{r}^h  \geq 0  \label{eq:rvleader5}
\end{align}\label{eq:tvleader}
\end{subequations}
\end{problem}
In Problem \ref{prb:rvleader}, the objective function (\ref{eq:rvleader0}) is  the system-level earnings, which is defined as the difference between the earnings of the operator from both HVs and AVs (the first term) and costs, including the operational cost for the passenger routes of HVs and AVs (the second term), the operational cost for the rebalancing routes of HVs and AVs (the third term), and the cost associated with passengers waiting to be matched with a driver (the fourth term), where $\psi$ represents a penalty for passenger waiting, which can be chosen, for example, as the VOT of passengers. Notice that Eq.(\ref{eq:rvleader0}) does not include the cost of time for HVs, since these vehicles have already participated in the system and would experience a fixed cost as the product between their VOT and their time in the system.  Constraints Eq.(\ref{eq:rvleader1}) and Eq.(\ref{eq:rvleader2}) represent the evolution of vehicle accumulation of AVs and HVs, respectively,  at each station.  Constraints Eq.(\ref{eq:rvleader3}) represent the evolution of waiting passengers with respect to each origin-destination pair, where $\epsilon\in [0,1]$ represents the probability that unserved passengers choose to remain in the system if currently unserved. Constraints Eq.(\ref{eq:rvleader4}) model the behaviors of HVs, where the specific form of function $\Phi(\cdot)$ is specified by the follower model as detailed below in Problem~\ref{prb:rvfollower}. Constraints Eq.(\ref{eq:rvleader5}) ensure that all variables are nonnegative. We next specify the follower model.
 \begin{problem}[Follower model in the MPC formulation] \label{prb:rvfollower}
For any $t\in\mathcal{T}_0$, given parameters $\bm{\phi}=\{\phi_i\}_{i\in\mathcal{N}}$ (modeling expected earnings), the expected earning rate $v$, compensations $\bm{c}_t$, remaining demand $\bar{\bm{q}}_t=\bm{w}_t + \bm{q}_t - \bm{x}_t^{a}$, vehicle availability $\bar{\bm{r}}_t=\bm{P}_t\bm{r}_t+\bm{A}^T\bm{x}_t^h$, where $\bm{A}^T$ represents the transpose of the adjacency matrix of graph $\mathcal{G}$, we derive the passenger flow of HVs at time step $t$ by solving the  optimization problem: 
\begin{subequations}
\begin{align}
\max_{\bm{x}^h} &\quad J_F^{\rm{TV}} = \sum_{(i,j)\in \mathcal{E}}(c_{ijt}-\phi_i+\phi_j-v\tau_{ij}-\sigma\delta_{ij})x_{ijt}^h \label{eq:rvfollower0}\\
\rm{s.t.}&\quad 0\leq x_{ijt}^h \leq \bar{q}_{ij},~(i,j)\in\mathcal{E} \label{eq:rvfollower1}\\
&\quad \sum_{j \in \mathcal{N}\backslash\{i\}}x_{ijt}^h \leq \bar{r}_{it} , ~i \in \mathcal{N}  \label{eq:rvfollower2} 
\end{align} \label{eq:rvfollower} 
\end{subequations}
\end{problem}

In Problem~\ref{prb:rvfollower}, the objective function Eq.(\ref{eq:rvfollower0}) 
is a proxy for the utility of HVs, where $\bm{\phi}$ and $v$ are parameters learned from data. Here $\phi_i$ is an estimate of expected earnings for vehicles located at station $i$, and as before $v$ represents system-wide average earning rate.    
Constraints Eq.(\ref{eq:rvfollower1}) ensure that the passenger flow is always nonnegative, and reflect the fact that HVs might decline a subset of assigned passenger requests. Constraints (\ref{eq:rvfollower2}) ensure that the resulting passenger flow does not violate the availability of HVs. 
Problem~\ref{prb:rvfollower} follows the intuition of the HV equilibrium analyzed in Section~\ref{sec:ssHV}.  More formally, we connect the definition of HV equilibrium as per Definition~\ref{def:HVEquilibrium} to Problem~\ref{prb:rvfollower} by establishing the following result.

\begin{proposition}[Connection of Problem~\ref{prb:rvfollower} to HV equilibrium] \label{prp:implication}
Given compensations $\bm{c}_t$ and  remaining demand $\bm{\bar{q}}_t$, there exists number of vehicles at stations $\bar{\bm{r}}_t=\{r_i\}_{i\in\mathcal{N}}$, expected earning starting from each station $\bm{\phi}$, and the system-wide expected earning rate $v$, such that the equilibrium of HVs defined in Definition~\ref{def:HVEquilibrium} is a solution to Problem~\ref{prb:rvfollower} with parameters $(\bm{c}_t,\bm{\bar{q}}_t,\bar{\bm{r}}_t,\bm{\phi},v)$. 
 \end{proposition}
 \begin{proof}
 See \ref{proof:implication}.
 \end{proof}
Proposition~\ref{prp:implication} indicates that the follower model as stated  Problem~\ref{prb:rvfollower}  can produce the same passenger flow as in the steady-state scenario, if the constraints of HV availability  Eq.(\ref{eq:rvfollower2}) are not active, i.e., the vehicle provision at each station is sufficient. In other words, the follower problem can be seen as an approximation to the equilibrium of HVs. Moreover, the expected earnings $\bm{\phi}=\{\phi_i\}_{i\in\mathcal{
N}}$  can also be obtained from the optimal dual variables of constraints Eq.(\ref{eq:HVOpt2}) in Problem~\ref{prb:HV1}, if the real data is limited.
   
Since the follower model as stated in Problem~\ref{prb:rvfollower} is a linear programming problem, the leader and follower models can be combined to yield a mixed  integer linear programming  (MILP). However, this approach is computationally expensive, especially for a large-scale city-level transportation network. To  improve scalability, we propose the following relaxation to the MPC model.  

\begin{problem}[Relaxation to the MPC formulation] \label{prb:combined}
Given a tuple $\{\bm{q}_t,\bm{p}_t,\bm{c}_t\}_{t\in \mathcal{T}_0}$ of predicted demand, compensations, and prices, as well as the initial conditions $(\bm{r}_0^a,\bm{r}_0^h,\{\bm{x}_t^h,\bm{y}_t^h,\bm{x}_t^a,\bm{y}_t^a\}_{t\in\mathcal{T}_{-}})$ and a weight parameter $\lambda\geq0$, the operator determines the passenger and rebalancing flow of AVs $\{\bm{x}_t^a,\bm{y}_t^a\}_{t\in\mathcal{T}_{0}}$ by solving the following optimization model~(\ref{eq:combined}).

\begin{subequations}
\begin{align}
\max_{\bm{x}^a, \bm{y}^a, \bm{w}, \bm{r}^a, \bm{r}^h} \quad & J^{\rm{TV}} = J_L^{\rm{TV}}(\bm{x}^a, \bm{y}^a, \bm{w}, \bm{r}^a, \bm{r}^h)+\lambda \sum_{t\in\mathcal{T}_0}J_{F,t}^{\rm{TV}}(\bm{x}^h) \label{eq:combined0}\\
\rm{s.t.} \quad  & \rm{Eq.(\ref{eq:rvleader1})-Eq.(\ref{eq:rvleader5}), Eq.(\ref{eq:rvfollower1})-Eq(\ref{eq:rvfollower2})}.\notag 
\end{align}\label{eq:combined}
\end{subequations}
\end{problem}
In Problem \ref{prb:combined}, the objective function combines the objective function of the leader model as stated in  Problem~\ref{prb:rvleader} and the  objective function of the follower model as stated in Problem~\ref{prb:rvfollower} via a weight parameter $\lambda\geq 0$. 
Such a parameter can be determined by the operator by using a  sensitivity analysis. We will indeed analyze the sensitivity to parameter $\lambda$ in Section~\ref{sec:simulation}. We also performed numerical experiments to evaluate the quality of the relaxation on small-scale  case studies (e.g., with 8 stations) with realistic problem data (e.g., in terms of system-level earnings, passenger acceptance rate, etc.) Results showed that the relaxation is on average able to find a solution with an optimality gap within 2\%--this motivates the use of Problem \ref{prb:combined} in our large-scale case studies. 

\subsection{Extension to general pick-up locations}
\label{sec:tvpickup}
In many existing MoD systems, vehicles can be matched with passengers 
even if they are not located within the same station as the passengers.  In this subsection, we extend the above framework to account for this possibility. We represent the passenger routes $\xi = (s,o,d) \in \mathcal{P}$ 
as the hyper-arc that vehicles at station $s\in\mathcal{N}$ take to pick up passengers at station $o\in\mathcal{N}$, and deliver them to their destination $d\in\mathcal{N}$. With such a notation,  we can extend Problem \ref{prb:combined} to Problem~\ref{prb:pickup} by replacing the variables for passenger flow. 

\begin{problem}[Extension to general  pick-up locations]\label{prb:pickup}
 Let $K$ be the length of planning horizon, $\mathcal{T}_0=\{t_0,t_0+1,~t_0+K-1\}$ be the set of time steps in the planning horizon from time step $t_0$, and $\mathcal{T}_{-}$ be the set of time steps prior to time step $t_0$. Given a tuple $\{\bm{q}_t,\bm{p}_t,\bm{c}_t\}_{t\in \mathcal{T}_0}$ of predicted demand, compensations, and prices, as well as the initial conditions $(\bm{r}_0^a,\bm{r}_0^h,\{\bm{x}_t^h,\bm{y}_t^h,\bm{x}_t^a,\bm{y}_t^a\}_{t\in\mathcal{T}_{-}})$, the operator determines the passenger and rebalancing flow of AVs $\{\bm{x}_t^a,\bm{y}_t^a\}_{t\in\mathcal{T}_{0}}$ by solving the  optimization problem
\begin{subequations}
\begin{align}
\max_{\substack{\bm{x}^a, \bm{y}^a,\bm{x}^h\\ \bm{w}, \bm{r}^a, \bm{r}^h}} ~ & J_L^{\rm{TV}} =  \sum_{t\in\mathcal{T}_0}\sum_{\xi=(s,o,d) \in \mathcal{P}} p_{odt}(x_{\xi t}^h+x_{\xi t}^a) 
- \psi \sum_{t\in\mathcal{T}_0}\sum_{\xi =(s,o,d)\in \mathcal{P}} \tau_{so}(x_{\xi t}^h+x_{\xi t}^a)
-\psi\sum_{t\in\mathcal{T}_0}\sum_{(i,j) \in \mathcal{E}}w_{ij}^t \notag \\
& - \sigma \sum_{t\in\mathcal{T}_0}\sum_{\xi =(s,o,d)\in \mathcal{P}} (\delta_{so}+\delta_{od})(x_{\xi t}^h+x_{\xi t}^a) - \sigma\sum_{t\in\mathcal{T}_0}\sum_{(i,j) \in \mathcal{E}}\delta_{ij}(y_{ijt}^a+P_{ijt}r_{i,t}^{h})  \notag \\
&+ \lambda \sum_{t\in\mathcal{T}_0}\sum_{\xi=(s,o,d)\in \mathcal{P}}(c_{odt}-\phi_s+\phi_d-v\tau_{so}-v\tau_{od})x_{\xi t}^h \label{eq:tvpickup0}\\
\rm{s.t.} ~  & r_{i,t+1}^{a}  = r_{it}^{a} + \sum_{j \in \mathcal{N}_i}y_{ji,t-\tau_{ji}}^{a} - \sum_{j \in \mathcal{N}_i}y_{ijt}^{a} + \sum_{\xi=(s,o,i) \in \mathcal{P}}x_{\xi,t-\tau_{so}-\tau_{oi}}^{a}-  \sum_{\xi=(i,o,d) \in \mathcal{P}}x_{\xi t}^{a},  ~i \in\mathcal{N},~t\in\mathcal{T}_0 \label{eq:tvpickup1}\\
& r_{i,t+1}^{h} = \sum_{j\in\mathcal{N}}P_{jit}r_{j,t}^{h} + \sum_{\xi=(s,o,i) \in \mathcal{P}}x_{\xi,t-\tau_{so}-\tau_{oi}}^{h}-  \sum_{\xi=(i,o,d) \in \mathcal{P}}x_{\xi t}^{h},~i \in\mathcal{N},~t\in\mathcal{T}_0 \label{eq:tvpickup2}  \\
& w_{ij,t+1} = \epsilon w_{ijt} + q_{ijt} - \sum_{\xi=(s,i,j) \in \mathcal{P}}(x_{\xi,t}^{a}+x_{\xi,t}^{h}),~(i,j) \in \mathcal{E},~t\in\mathcal{T}_0 \label{eq:tvpickup3} \\
&\bm{x}^a, \bm{y}^a, \bm{x}^h, \bm{w}, \bm{r}^a, \bm{r}^h  \geq 0 \label{eq:tvpickup4}
\end{align}\label{eq:tvpickup}
\end{subequations}
\end{problem}
In Problem~\ref{prb:pickup}, the objective function Eq.(\ref{eq:tvpickup0}) is to maximize system-level earnings, defined as the difference between the total prices of the requests (the first term) and costs, including the passenger waiting cost associated with the pickup process (the second term), the cost for passengers waiting to be matched with a driver (the third term), the operational costs for passenger routes (the fourth term), the operational costs for rebalancing routes (the fifth term), and a penalty term (the sixth term) that characterizes the objective function of the follower problem (Problem~\ref{prb:rvfollower}) to account for the self-interested behavior of HVs. Notice that the differences between Eq.(\ref{eq:tvpickup0}) and Eq.(\ref{eq:rvleader0}) are 1) we explicitly penalize the passenger pickup time, and 2) we allow positive passenger pickup times within the same region (i.e., $\tau_{ii}>0,~i\in\mathcal{N}$). Constraints  Eq.(\ref{eq:tvpickup1}) and Eq.(\ref{eq:tvpickup1}) represent the evolution of AVs and HVs, respectively, and Constraints  Eq.(\ref{eq:tvpickup3}) describes the evolution of passengers. Eq.(\ref{eq:tvpickup4}) ensures that all variables are nonnegative. 

To summarize this section, we have presented a Stackelberg game-based MPC approach to control the MoD systems with mixed fleet in real-time. The leader problem solved by the operator is formulated as Problem~\ref{prb:rvleader} to optimize the system-level earnings, and the follower problem approximating the behavior of HVs is formulated as Problem~\ref{prb:rvfollower}. The Stackelberg game based MPC approach is approximated as Problem~\ref{prb:combined} to improve scalability, and is extended to consider general pick-up locations as Problem~\ref{prb:pickup}.

\section{Case Study and Numerical Evaluation}

\label{sec:simulation}
\begin{figure}[htbp]
\center
\includegraphics[width=0.45\textwidth]{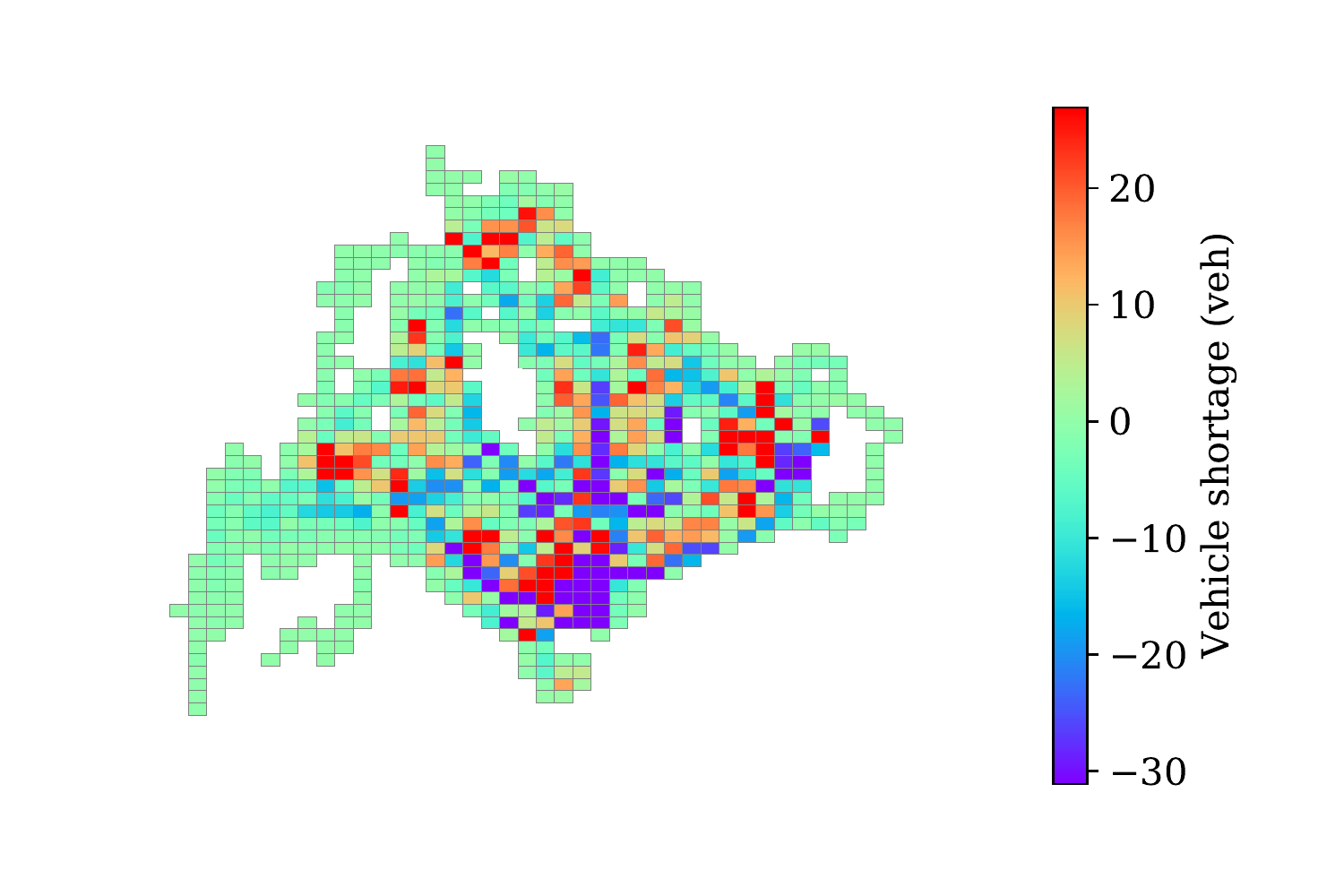}
\caption{Vehicle imbalances in Singapore as a function of locations, where the red blocks represent the locations with vehicle shortage and the blue blocks represent the locations with vehicle surplus.   }
\label{fig:network}
\end{figure}

\subsection{Experiment design and data}

We perform two case studies: a steady-state analysis with a static demand and a real-time simulation with a realistic demand. The steady-state analysis determines the planning variables (e.g., trip prices and compensations) and conceptually quantifies the long-term benefits of AVs on system-level earnings. The real-time simulation, on the other hand, validates the proposed MPC approach by evaluating the system-level earnings in scenarios with given prices and compensations (derived using the method in Section~\ref{sec:ss} or from real data).  

In both case studies, we simulate a typical weekday morning peak between 8 a.m. and 10 a.m. in Singapore, where the road network is discretized into 126 blocks, each with an area of 2.4\,km$\times$2.4\,km.  
The Singapore dataset is provided by Grab Holding Inc., which includes the passenger demand, the average travel times between OD station pairs, the pick up times, and fares.  
Figure~\ref{fig:network} shows the studied area and the vehicle imbalances in the road network in the scenario where all vehicles are human-driven. Here, vehicle imbalances are characterized as the difference between the number of vehicles needed and the vehicle supply. It is clear that vehicles are oversupplied in some of the regions but insufficient in others. The VOT of passengers is assumed to be the average salary of Singapore, i.e., 0.45\,SGD/min \citep{MoMSingapore2020}. The operational cost of vehicles is chosen to be 0.15\,SGD/km \citep[estimated from][]{BoeschBeckerEtAl2018}. The expected earning rate of HVs is calibrated to be $v=0.25$\,SGD/min. 


Simulations are run in Python on a laptop computer with Intel Core i5-8265U and 8 GB memory. The optimization models are solved using CPLEX.

\subsection{Steady-state analysis of the MoD system with mixed fleets}

We perform a steady-state analysis to quantify the long term benefits of AVs. 
In practice, such a steady-state analysis can be performed for different periods of the day (morning commute, afternoon, evening commute, etc) where the demand and supply are relatively static. 
Here, we use the average transportation demand and travel times within the studied period of 8 a.m. to 10 a.m. in the Singapore dataset. The CDF of the willingness to pay $F_{ij}$ is assumed to be a uniform distribution between a minimum price and a maximum price calibrated from the Singapore dataset as a function of travel times and travel distances.
We evaluate scenarios with various AV and HV provision (i.e., the upper bounds on the fleet sizes). The upper bound on the fleet size for AVs, $N^a$, varies between 0 and 24,000 vehicles, and the upper bound on the fleet size for HVs, $N^h$, varies between 0 and 16000 vehicles.  
Prices, compensations, and the passenger and rebalancing flows of vehicles of both classes are calculated based on Theorem~\ref{thm:mixed}. Recall that these compensations can encourage HVs to act in a system-optimal manner should they participate in the system.  
The results are shown in Figure~\ref{fig:ss}. 

\begin{figure}[t]
\centering
\begin{subfigure}{0.48\textwidth}
    \includegraphics[width=\textwidth]{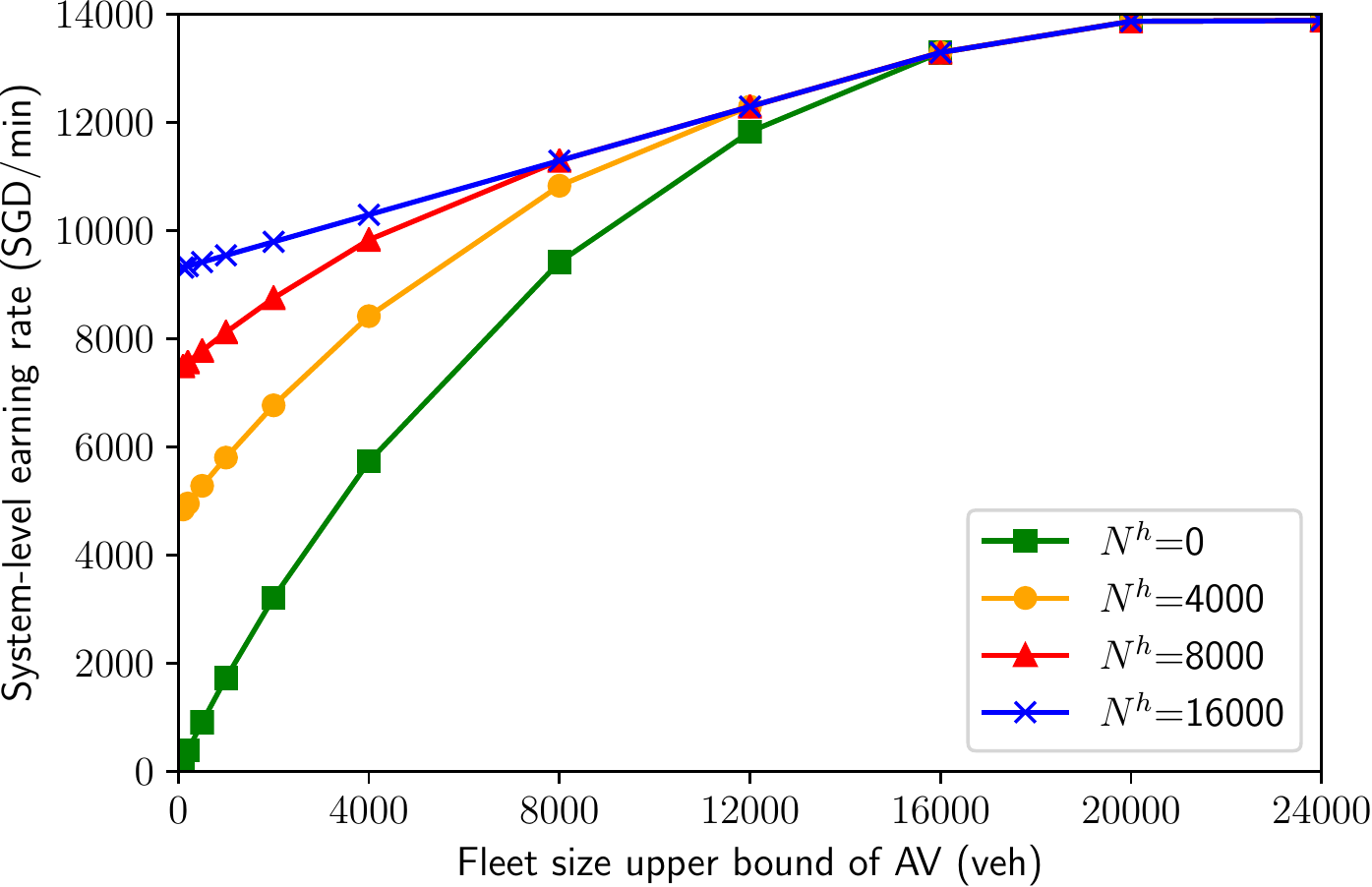}
    \caption{System-level earning rate.} \label{fig:ss_social}
\end{subfigure}
\begin{subfigure}{0.48\textwidth}
    \includegraphics[width=\textwidth]{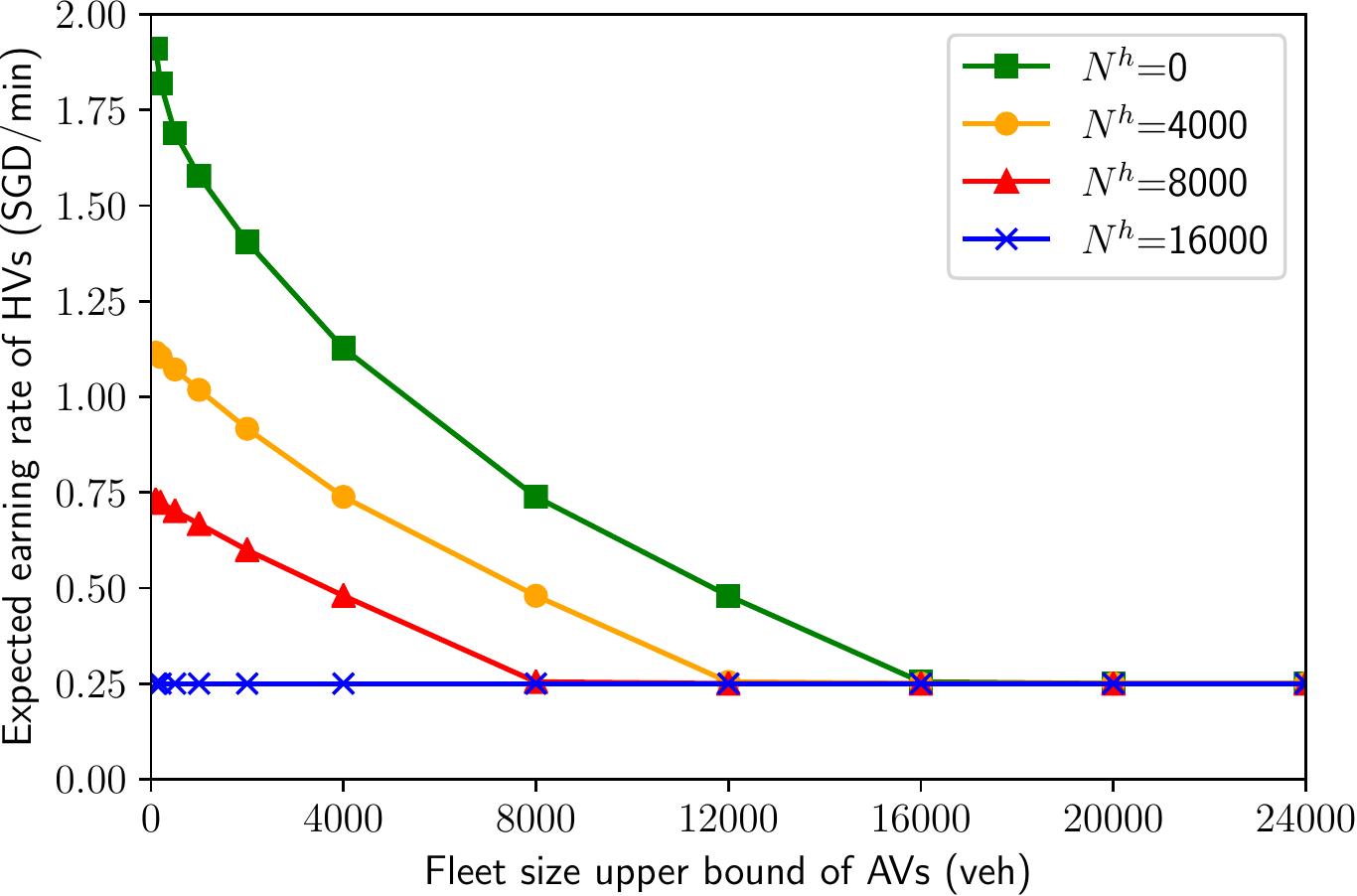}
    \caption{Expected earning rate of HVs.} \label{fig:ss_mu}
\end{subfigure}  
\begin{subfigure}{0.48\textwidth}
    \includegraphics[width=\textwidth]{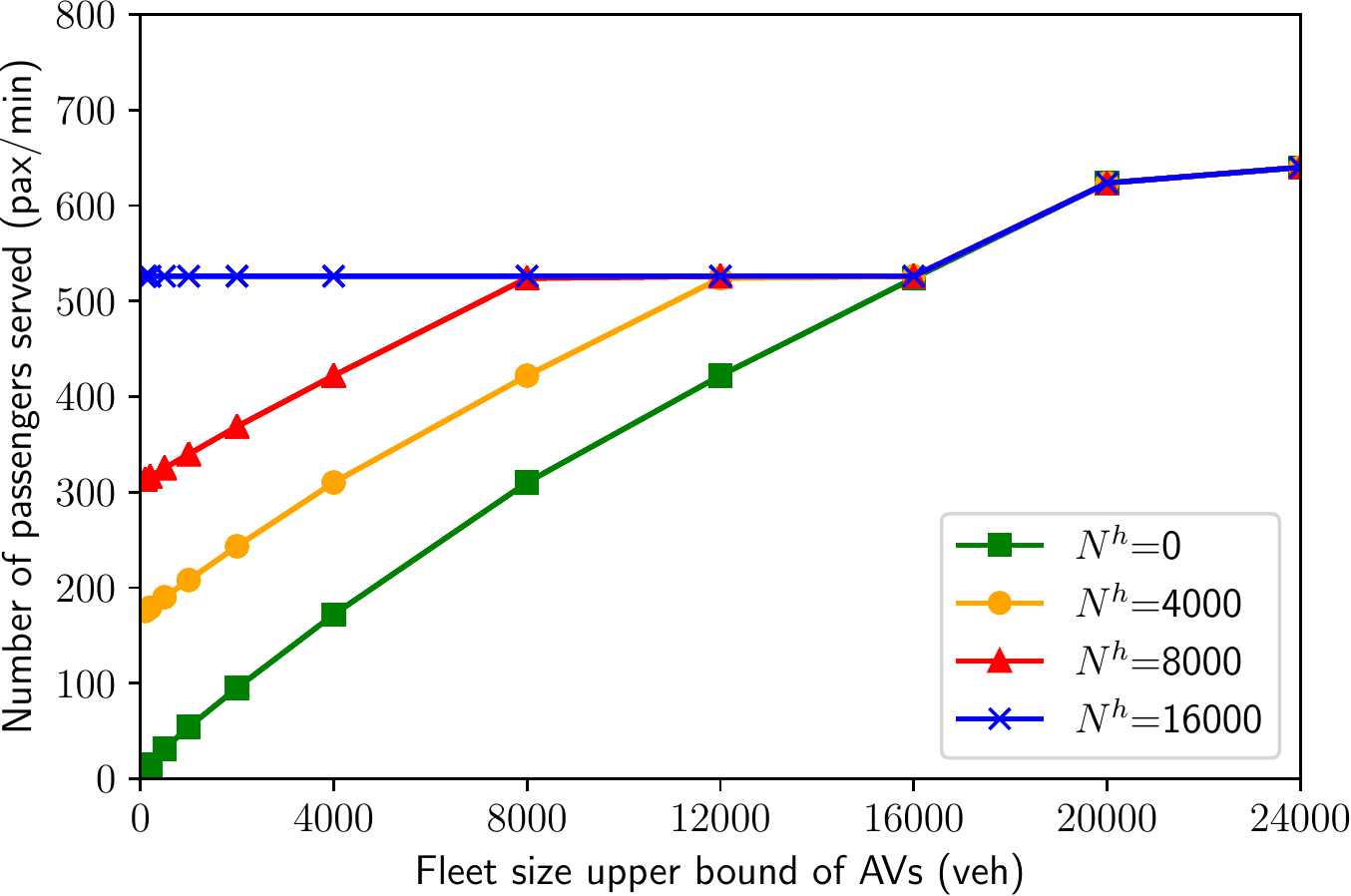}
    \caption{Number of passengers served} \label{fig:ss_demand}
\end{subfigure}  
    \caption{Steady-state analysis in scenarios with various AV ($N^a$ ) and HV ($N^h$) provision.}
    \label{fig:ss}
\end{figure}

Figure~\ref{fig:ss_social} shows the value of having more AVs by evaluating the resulting system-level earning rate (i.e., the system earnings per time step) in scenarios with various upper bounds on the fleet size of HVs (i.e., $N^h$). We can see that as the fleet size of AVs increases, the system-level earning rate increases and the marginal benefits are more evident if the fleet size of AVs is low. Moreover, restricting the number of HVs in the system would  reduce the system-level earnings. 

Figure~\ref{fig:ss_mu} shows how the introduction of AVs would impact the earning rate of HVs in scenarios with various upper bounds on the fleet size of HVs.  We can see that the earning rate of HVs is in general reduced by the introduction of AVs. This is because the AVs have lower costs and would thus replace HVs. 
We can also see that by imposing a stricter bound on the number of HVs in the system (i.e., with smaller $N^h$), the earning rate of HVs can be improved since fewer HVs would be competing for requests.  

Figure~\ref{fig:ss_demand} shows the evolution of the number of served passengers per time step with the introduction of AVs, in scenarios with various upper bounds on the fleet size of HVs. We can see that, in general, the system is able to serve more passengers by having more AVs in the system. However, we notice that the number of served passengers remain constant in some scenarios (e.g., if the upper bound of HVs is 16,000 vehicles). Comparing with Figure~\ref{fig:ss_mu}, we can see that in these scenarios, the earning rate of HVs equal the VOT of HVs, and thus every AV entering the system will replace the routes of a HV since it has lower costs, and the number of served passengers per time step would remain constant in these scenarios.

\subsection{Performance of the proposed Stackelberg game-based MPC approach}\label{sec:performance}

The real-time simulation aims to validate the proposed MPC approach. We take the real-time transportation demand, travel times, and trip fares directly from the Singapore dataset. We first evaluate a scenario where the proposed MPC approach has perfect information on demand, travel times, and the rebalancing probabilities, and then test the robustness of the MPC approach against the prediction errors in these variables. The total fleet size is 14,000 vehicles, which is assumed to be fixed and time-invariant throughout the studied time period.  We vary the share of AVs in the mixed fleet to be between 0 and 14,000 vehicles. The length of a time step is chosen as 1\,min. At the beginning of each time step, passengers make requests, and idle vehicles located at the same or adjacent blocks can be assigned to the requests. 
Notice that the pickup times are positive even if the passengers are picked up by vehicles located at the same station. 
This is to account for the size of the blocks.
 The planning horizon for the proposed MPC is chosen as 15\,min as a trade-off between computational efficiency, prediction accuracy, and system performance. The weight parameter $\lambda$ is set via sensitivity analysis equal to be 6.

We evaluate the proposed Stackelberg game-based MPC approach by comparing the following three approaches.
\begin{itemize}
   \item Baseline approach, defined as an MPC approach that coordinates AVs assuming fully compliant HVs, i.e., without accounting for the strategic interactions with the HVs  (Problem~\ref{prb:pickup} with $\lambda=0$). Prices, demand, and compensations are taken from the Singapore dataset (as opposed to be computed through the steady-state planning model). 
   \item Stackelberg game-based MPC with actual compensations. Prices, demand, and compensations  are taken from the Singapore dataset, and the control actions are obtained by solving Problem~\ref{prb:pickup} with the $\lambda=6$ that provides the optimal system-level earnings in the sensitivity analysis. 
    \item Stackelberg game-based MPC with compensations set by using the steady-state formulation.  Prices and demand are taken from the Singapore dataset to facilitate comparison, whereas the compensations are set by using the dual variables associated with constraint Eq.(\ref{eq:AVDemand2}) (see Section~\ref{sec:ssMixed} for more details). The control actions are obtained by solving Problem~\ref{prb:pickup} with  $\lambda=6$.  
\end{itemize}
These approaches  are evaluated  in scenarios with different fleet sizes of AVs to illustrate the value of introducing AVs in the MoD system. The comparison of the first two approaches sheds light on the value of considering the interactions between AVs and HVs, while the comparison of the later two approaches demonstrates the value of optimizing the compensations.  
In these comparisons, we evaluate  the system-level earnings in general, but also look into more detailed performance criteria, e.g., operator's profit, average passenger waiting time, passenger acceptance rate, vehicle utilization, and the empty vehicle  kilometers traveled. Here, since we assume passengers do not wait to be assigned, the passenger waiting time refers specifically to the time that a passenger waits to be picked up by the matched vehicle. The passenger acceptance rate refers to the percentage of passengers that are successfully matched with a vehicle, the vehicle utilization represents the percentage of time a vehicle is occupied by a passenger, and the empty vehicle  kilometers traveled are defined as the total distance traveled by vehicles to pick up passengers or rebalance themselves. 
The results are summarized in Figure~\ref{fig:performance}. 
 \begin{figure}[htbp]
    \centering
    \begin{subfigure}{0.45\textwidth}
    \includegraphics[width=\textwidth]{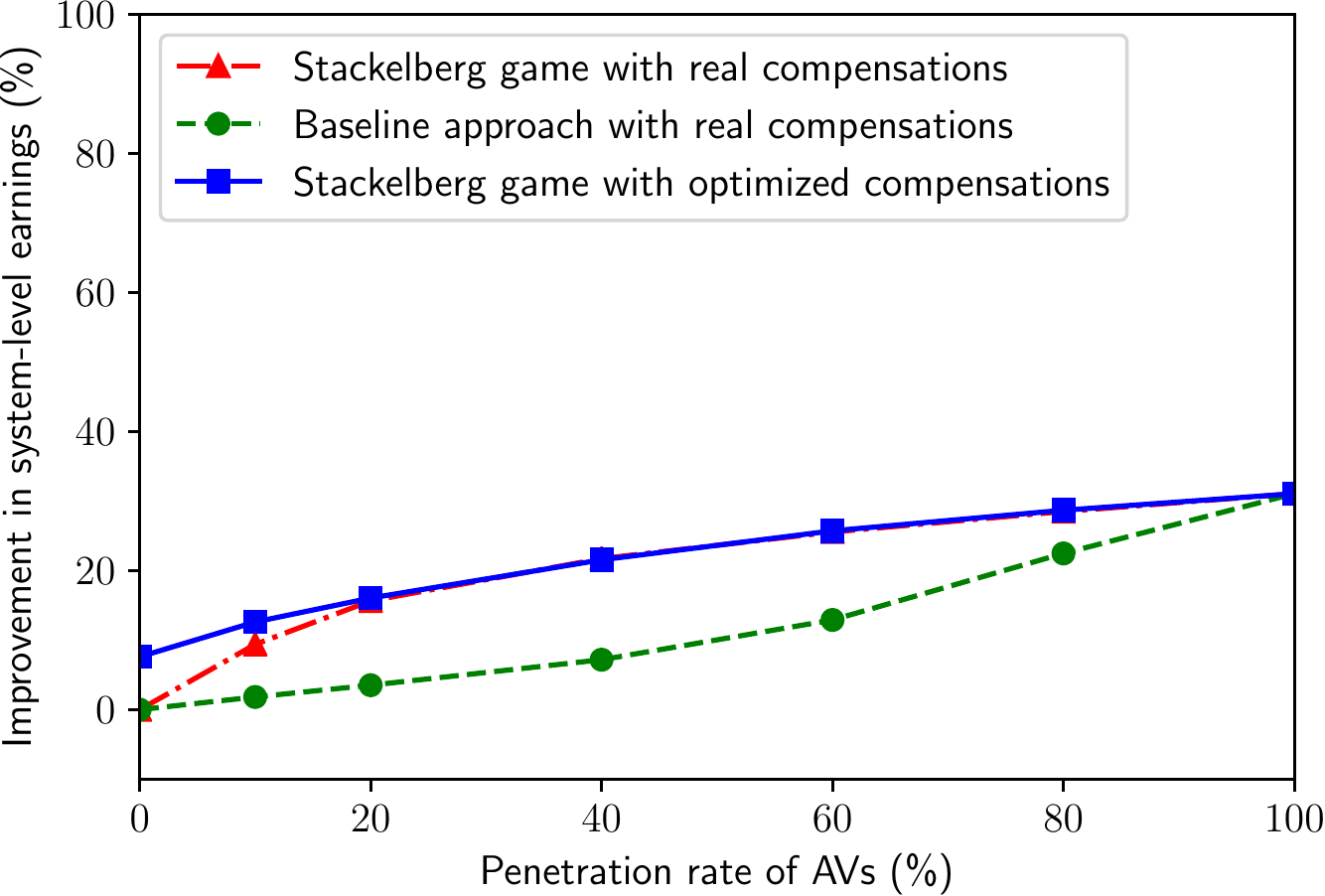}
    \caption{Improvement in system-level earnings} \label{fig:social}
    \end{subfigure} \quad \quad
    \begin{subfigure}{0.45\textwidth}
    \includegraphics[width=\textwidth]{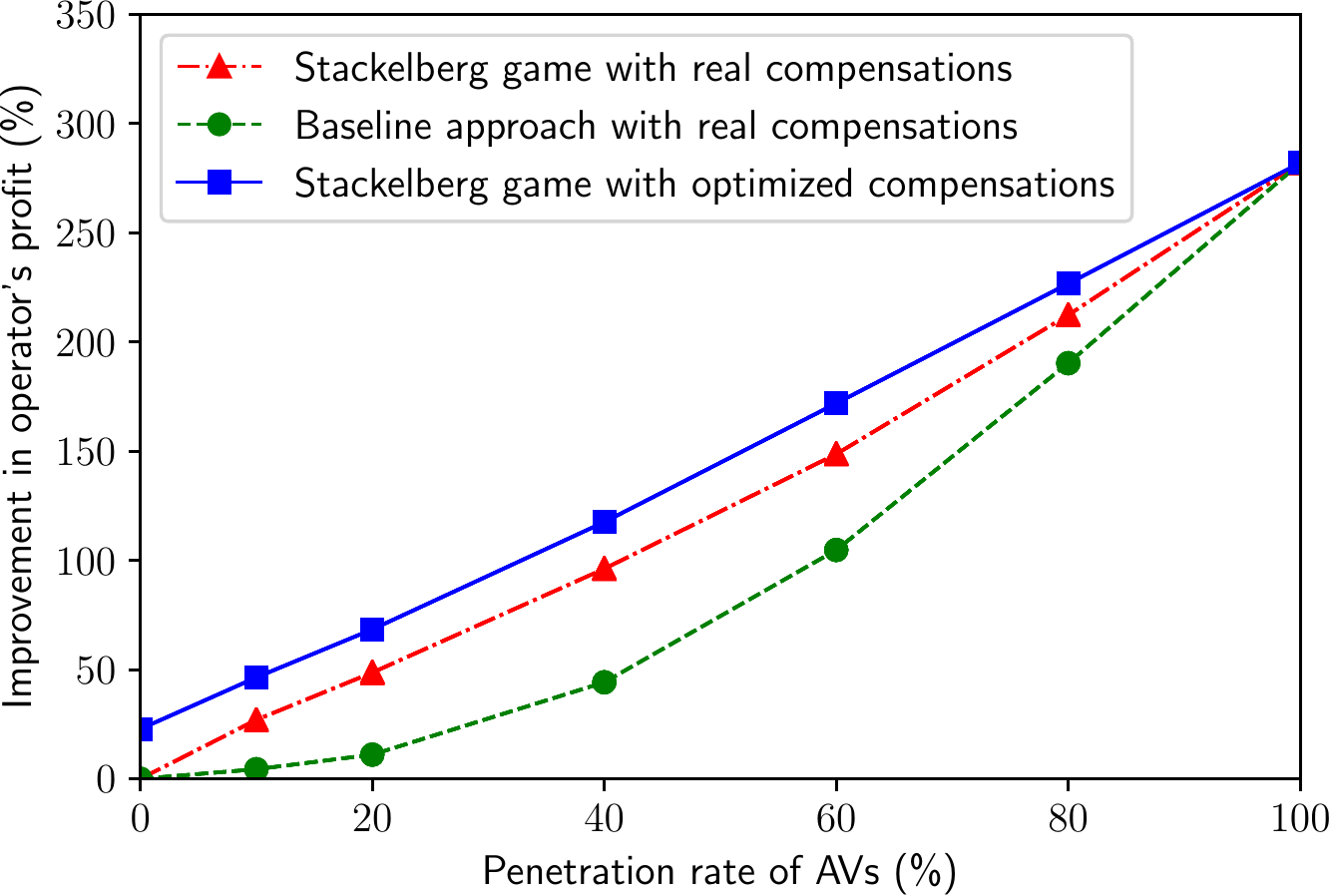}
    \caption{Improvement in operator's profit} \label{fig:operator_profit}
    \end{subfigure}
    \begin{subfigure}{0.45\textwidth}
    \includegraphics[width=\textwidth]{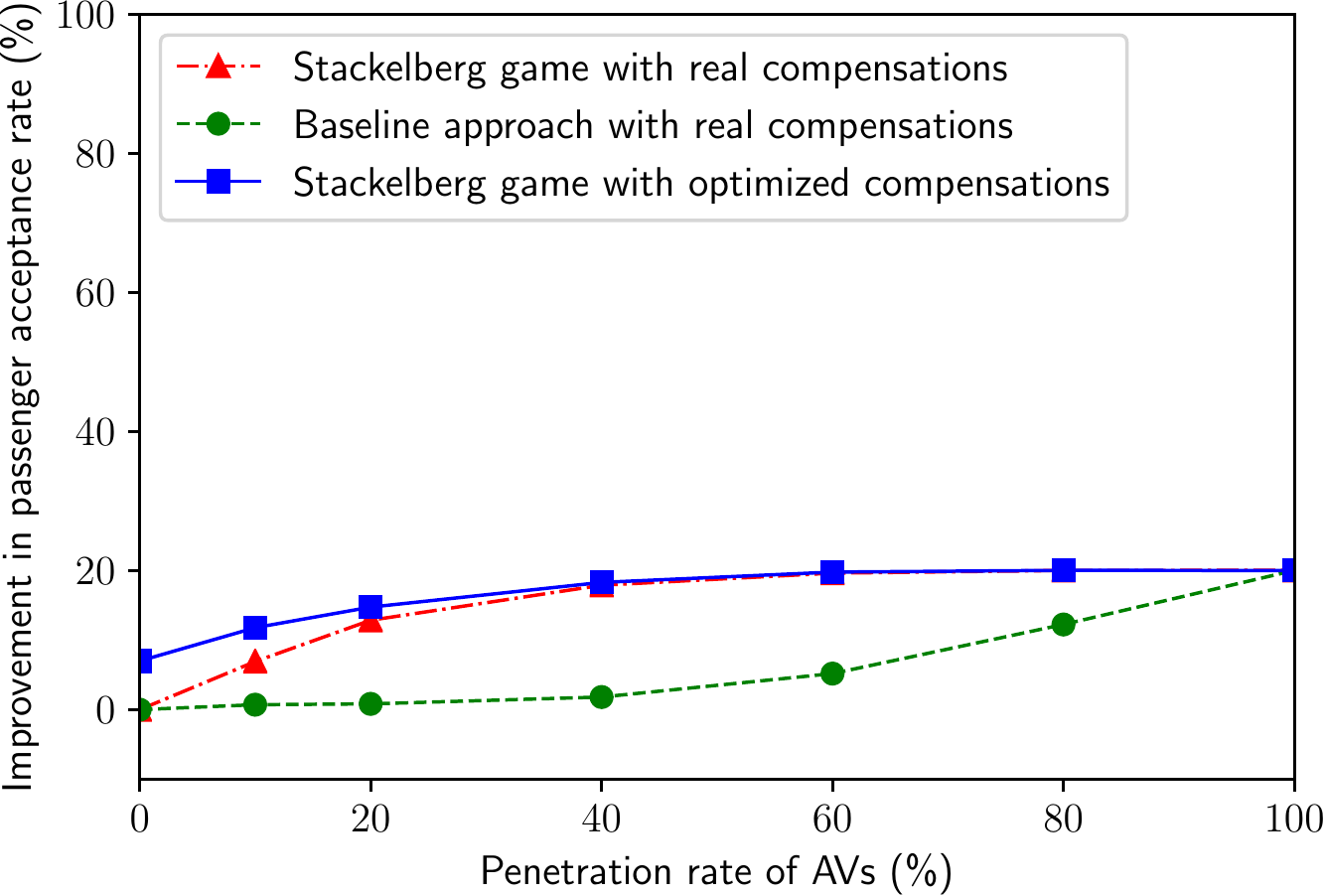}
    \caption{Improvement in passenger acceptance rate} \label{fig:pax_acceptance}
    \end{subfigure}\quad \quad
    \begin{subfigure}{0.45\textwidth}
    \includegraphics[width=\textwidth]{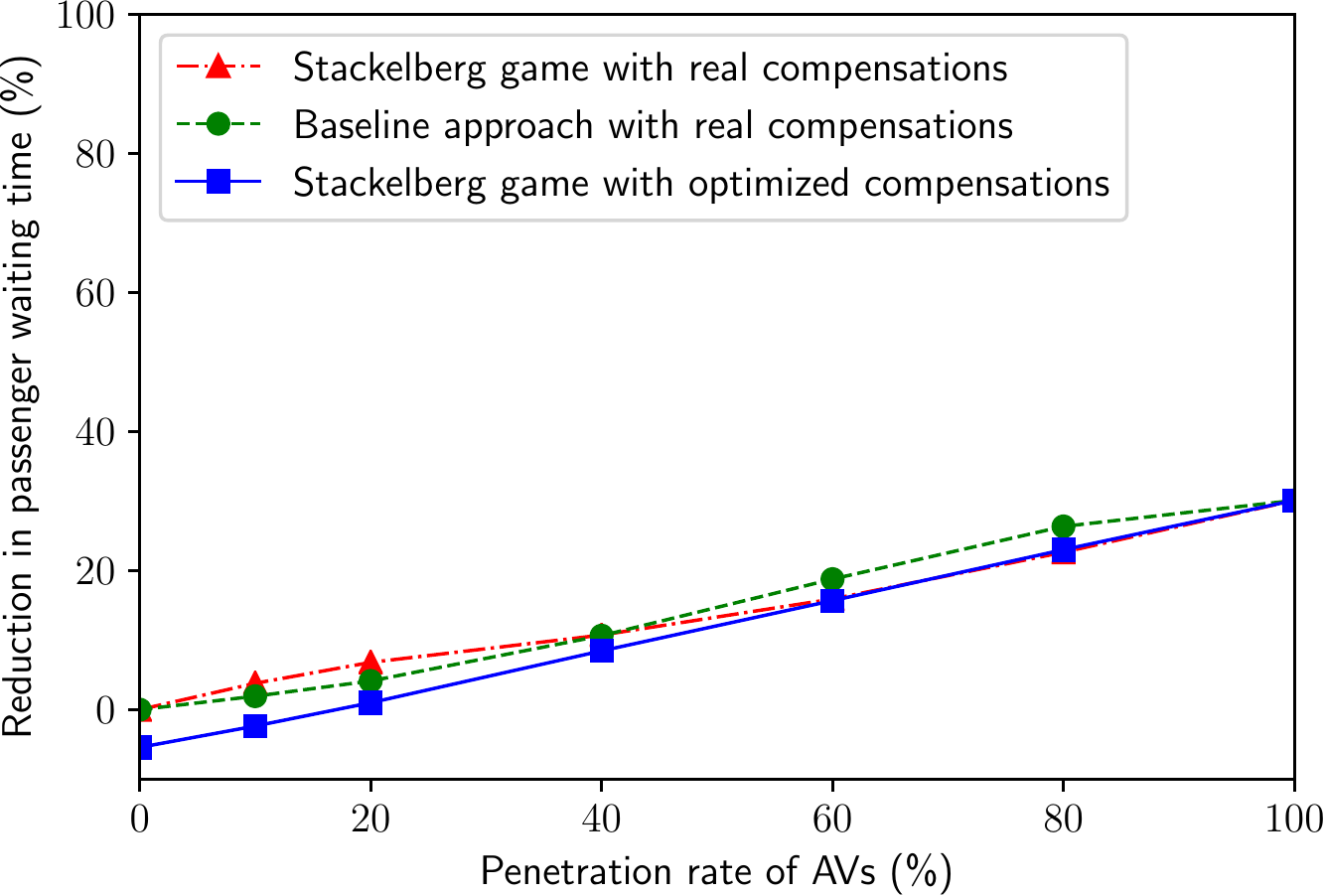}
    \caption{Reduction in passenger waiting time} \label{fig:pax_wait}
    \end{subfigure}
    \begin{subfigure}{0.45\textwidth}
    \includegraphics[width=\textwidth]{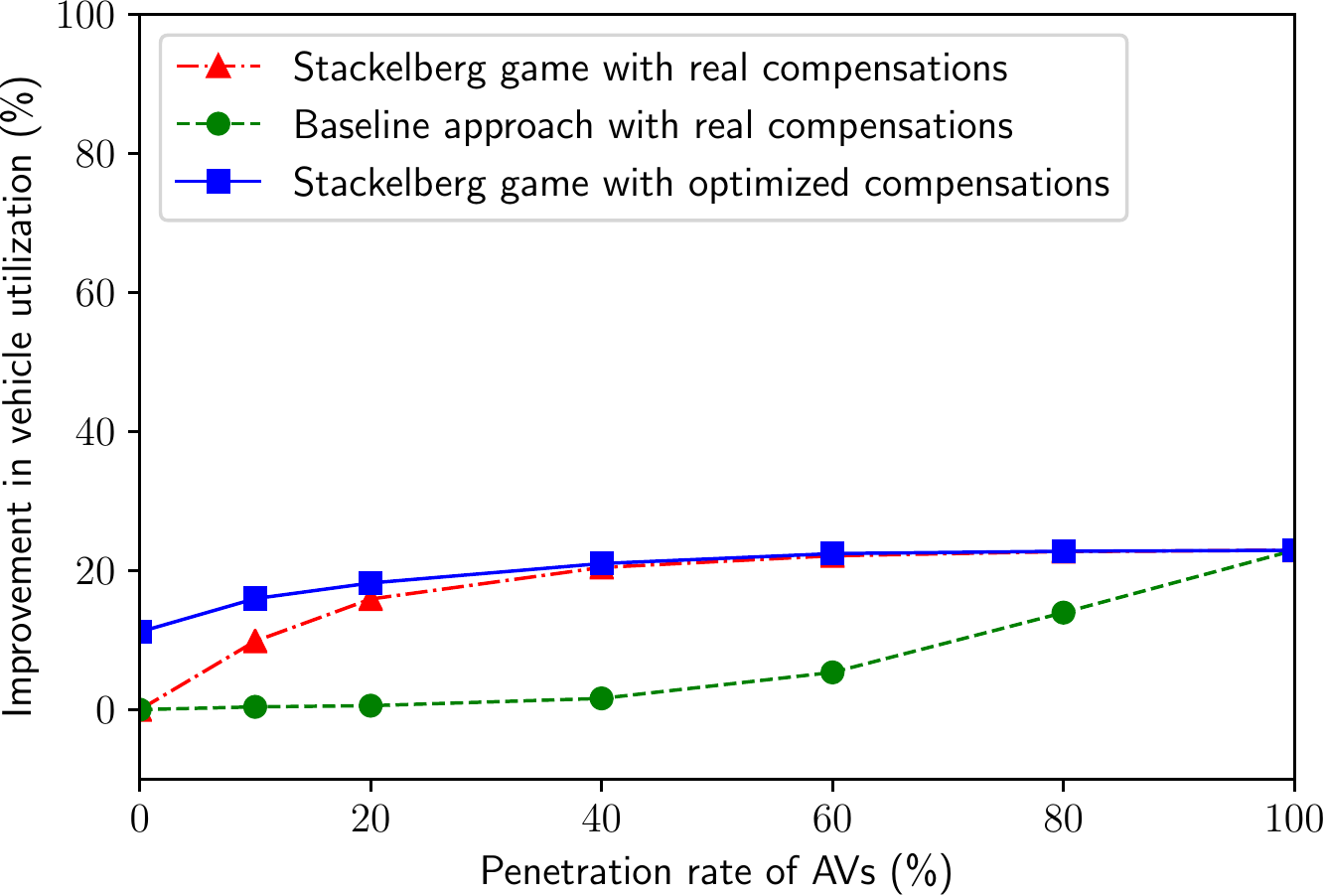}
    \caption{Improvement in vehicle utilization} \label{fig:veh_util}
    \end{subfigure}\quad \quad
    \begin{subfigure}{0.45\textwidth}
    \includegraphics[width=\textwidth]{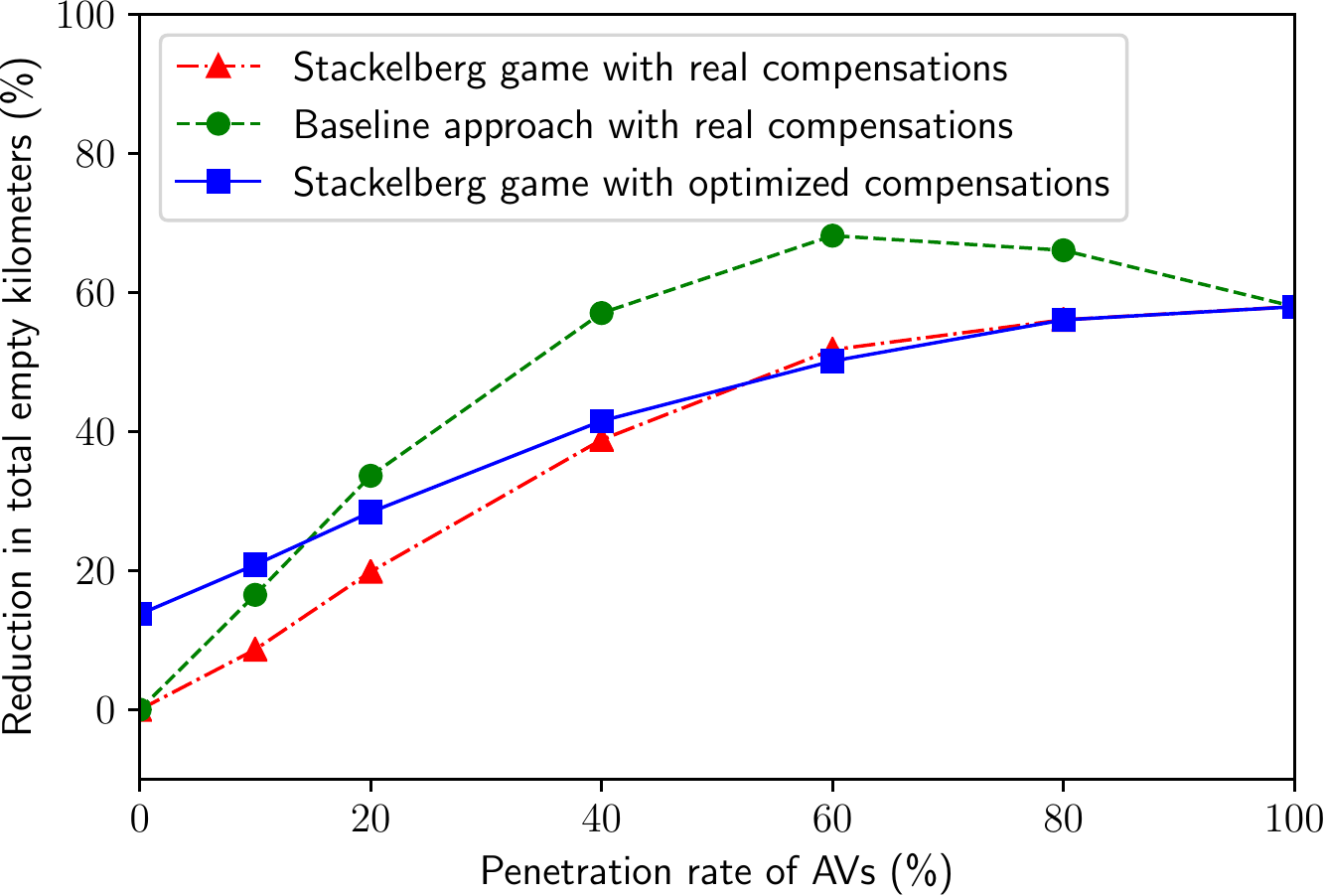}
    \caption{Reduction in total empty kilometers traveled} \label{fig:veh_cost}
    \end{subfigure} 
    \caption{Performance of the Stackelberg game based MPC approach compared to the scenario where all vehciles are HVs. }
    \label{fig:performance}
\end{figure}

\paragraph{Value of introducing AVs} We show the value of introducing AVs by comparing the system performance in scenarios with various penetration rates of AVs. 
We can see from Figure~\ref{fig:performance} that the MoD system can be improved significantly by replacing HVs with AVs by using all three approaches. 
Specifically, we can improve the system-level earnings by up to 32\%, reduce the passenger waiting time by up to 30\%, increase the passenger acceptance rate by up to 20\%, improve the vehicle utilization by up to 21\%, and reduce the empty kilometers traveled by up to 67\%.  
Moreover, the benefits of AVs is more significant at the early stage of deployment. In fact, the system-level earnings can be improved by 13\% by replacing 10\% of vehicles with AVs, and more than 97\% passengers can be served if the AV fleet size is larger than 40\%.  This shows that promising value of AVs in improving MoD services.

\paragraph{Value of considering the interactions between AVs and HVs} We show the value of considering the interactions between AVs and HVs by comparing the baseline approach (green dashed lines) with the Stackelberg game-based MPC approach  with actual compensations (red dotted lines).  We can see from Figure~\ref{fig:performance} that the proposed Stackelberg game-based MPC can improve system-level earnings by 12\% by considering the interactions between HVs and AVs. The benefit of considering such interactions is especially evident in scenarios with moderate penetration rates of AVs (i.e., 20\% -- 60\%). This is because the interactions between AVs and HVs are more intense in these scenarios. We also observe that the benefits of considering such interactions lie mainly in the improvement in the passenger acceptance rate. There are two reasons for this. First, the Stackelberg game-based MPC can predict the behaviors of HVs more accurately, and can match AVs to the passengers that HVs are not willing to serve. Second, by modeling the system more accurately, the Stackelberg game-based MPC is able to coordinate AVs in a more efficient manner, and thus serve more passengers. The passenger waiting time and empty kilometers traveled, however, may not be necessarily improved. This is expected because the Stackelberg game-based MPC can coordinate AVs to take the passengers with relatively low values and long pickup times that would otherwise not be taken by the HVs. Accepting these passengers may increase the average passenger waiting time and the total empty kilometers traveled. 

\paragraph{Value of optimizing the compensations.} We show the value of optimizing compensations by comparing the Stackelberg  game-based MPC approach with actual compensations (red dotted lines) and optimized compensations (blue solid lines). We can see from Figure~\ref{fig:performance} that by optimizing the compensations, we can improve the system-level earnings by up to 8\%, and the improvement is more evident when the penetration rate is low. This is reasonable, since  compensations aim to drive the behaviors of HVs to the system optimal actions and would be more beneficial if there were more HVs in the system. Similar trends can be observed for the passenger acceptance rate, vehicle utilization, the operator's profit, and the total empty kilometers traveled. Notice that the passenger waiting time may not be improved by optimizing compensations. One potential reason is that the passenger waiting time does not contribute significantly to the objective functions, since the value is typically not too large.    

\subsection{Robustness to prediction errors}
The proposed MPC approach relies on the prediction of several 
variables: the future travel demand,  the rebalancing flow of HVs, and the travel times. 
We analyze how the prediction errors in these variables would affect the performance of the proposed MPC approach. 
To this end, we allow these variables to be stochastic in the simulation, but use the mean values in the proposed MPC to calculate the actions for AVs. 
Specifically, the travel demand is assumed to follow a Poisson distribution.
The number of rebalancing HVs from each station is assumed to follow a multinomial distribution with the given rebalancing probabilities as the parameters. The prediction error in the travel times is assumed to follow a Gaussian distribution with mean 0 and a standard deviation of 10\% and 20\% of the value, respectively, to simulate scenarios with moderate and large noises in travel times. 
Notice that the MPC can obtain the true travel demand at the current time step from passenger requests, but uses the mean value of the Poisson distribution to predict the future demand.
  We derive the resulting system-level earnings in scenarios with different levels of prediction errors and different fleet sizes of AVs, each with 5 random seeds. The results are as shown in Figure~\ref{fig:robustness}. 

\begin{figure}[htbp]
    \centering
    \includegraphics[width=0.5\textwidth]{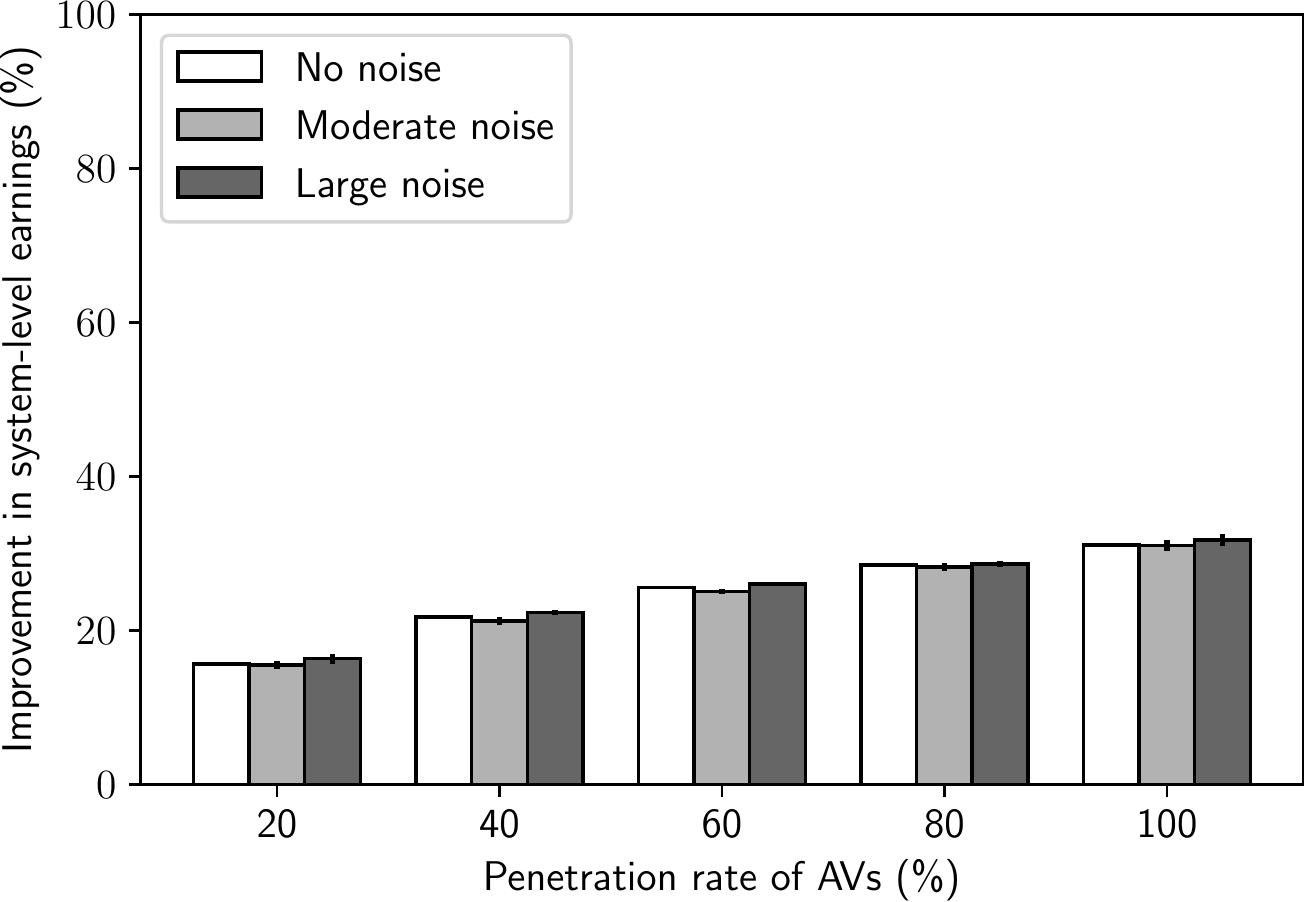}
    \caption{Robustness to prediction noises, where the vertical axis represents the improvement in system-level earnings compared to the scenario with all HVs; moderate noises and large noises represent the scenarios with 10\% and 20\% errors in travel time prediction, respectively; and the error bars represent the standard deviation in the improvement across different random seeds. }
    \label{fig:robustness}
\end{figure}

Figure~\ref{fig:robustness} shows that the proposed Stackelberg game-based approach is quite robust to these prediction errors. This is expected for three reasons. First, we expect MPC to be robust against a certain level of prediction errors, due to its built-in feedback mechanism. Second, transportation requests at the current time step are submitted by the passengers and are accessible to the operator. Hence, the controller has perfect knowledge of the demand at the current time step, which is used to effectively match passengers with drivers. 
Third, the knowledge of transportation demand and rebalancing probabilities  enables MPC to frequently rebalance AVs to blocks with vehicle shortages.

\subsection{Sensitivity to model parameters}

\begin{figure}[htbp]
    \centering
    \begin{subfigure}{0.45\textwidth}
    \includegraphics[width=\textwidth]{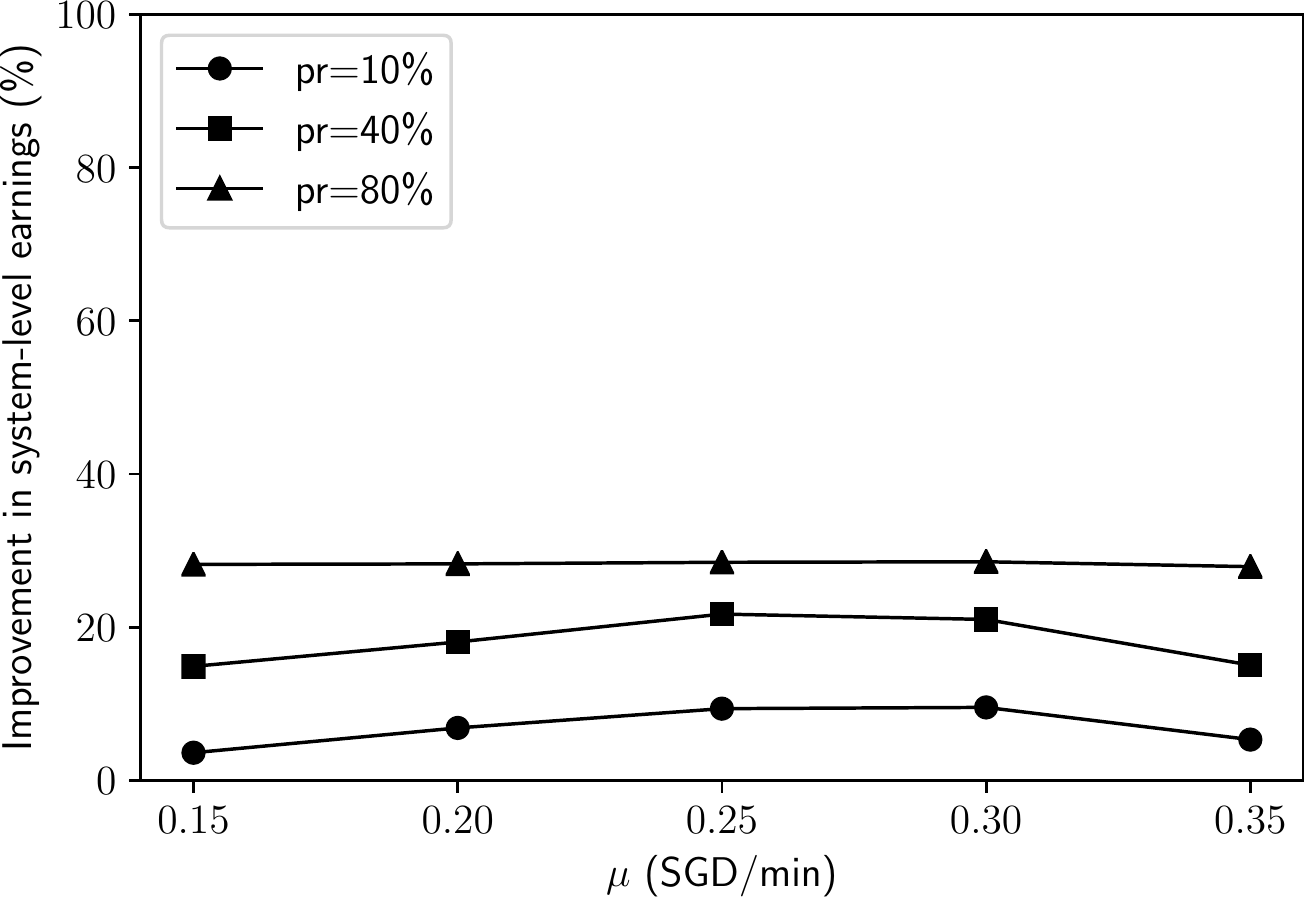}
    \caption{Sensitivity to the assumed VOT of HVs.}
    \label{fig:sa_vh}
    \end{subfigure}\quad\quad
    \begin{subfigure}{0.45\textwidth}
    \includegraphics[width=\textwidth]{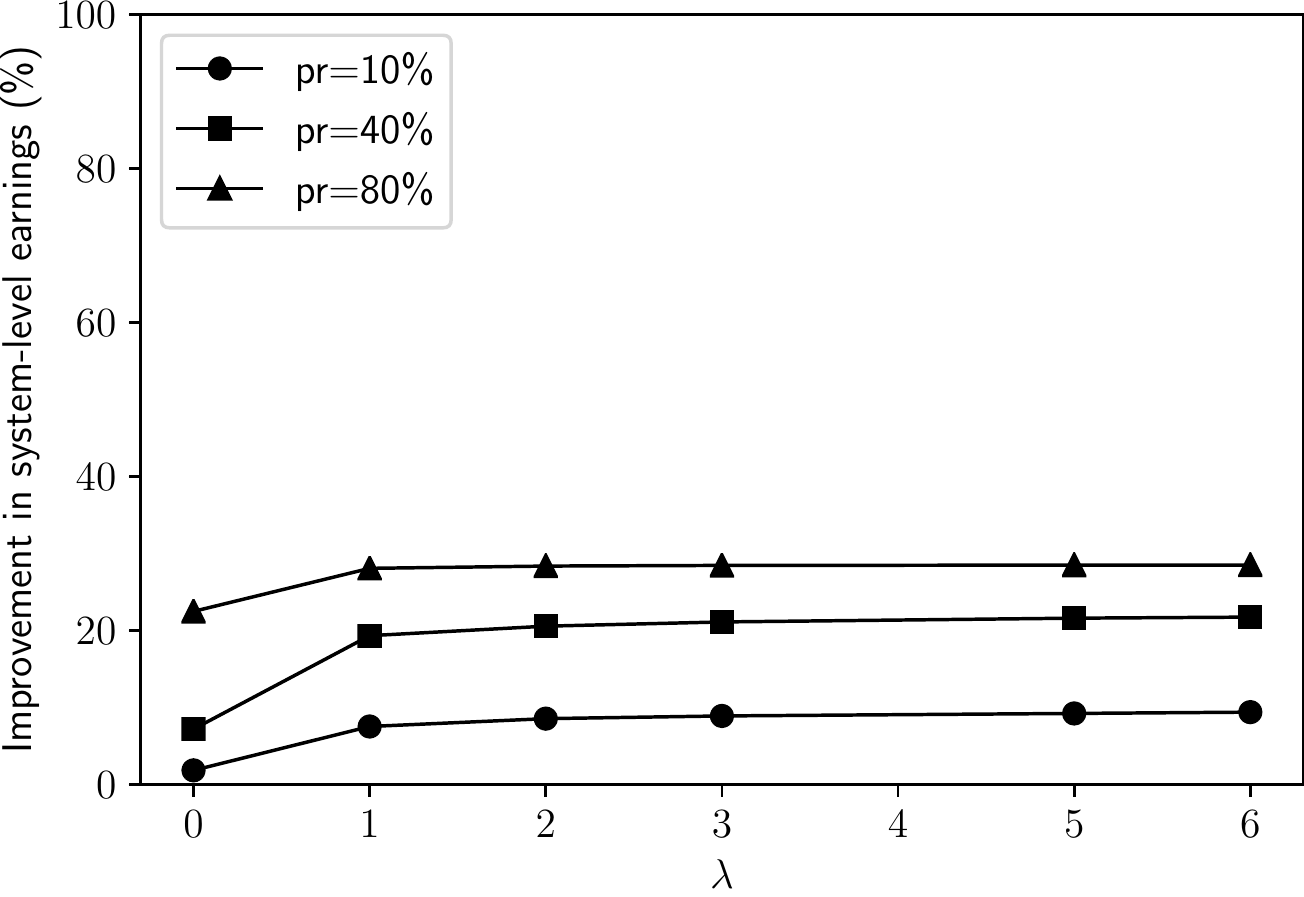}
    \caption{Sensitivity to the weighting parameter $\lambda$.}
    \label{fig:sa_lambda}
    \end{subfigure}
    \caption{Sensitivity analysis for the Stackelberg game based MPC approach in scenarios with various penetration rates of AVs, where the prices, demand, and compensations are taken from the Singapore dataset.  }
    \label{fig:sa}
\end{figure}

We analyze the sensitivity of the proposed 
MPC to the model parameters, i.e., the assumed VOT of HV drivers $\mu$ and the weight parameter $\lambda$ in Problem~\ref{prb:pickup}, to evaluate the impact on the performance of the proposed MPC if these parameters deviate from the values that yield the optimal system-level earnings. 
To this end, we employ a local sensitivity analysis method, i.e., the One-at-a-Time method by varying one parameter at each time  in the MPC and fixing the other parameter. 
Such a local sensitivity analysis suffices for practical applications, because we do not expect the parameter values
to deviate significantly from the optimal values. Specifically, we vary the assumed $\mu$ between 0.15\,SGD/min to 0.35\,SGD/min in Problem~\ref{prb:pickup}, while keeping $\mu=0.25$ in the simulation. The value of $\lambda$ varies between 0 and 6. 
The results are illustrated in Figure~\ref{fig:sa}. 

We can see that the performance of the proposed MPC is not sensitive to either parameter, as long as the VOT of HV drivers $\mu$ is within a reasonable range (i.e., between 0.2\,SGD/min and 0.3\,SGD/min) and the weight parameter $\lambda$ is greater than 3.  The performance of the proposed MPC would drop if the parameters are outside these ranges. If $\mu$ lies too far from the true value, the MPC may not model the system accurately, and therefore might sacrifice system performance.  On the other hand, if $\lambda$ is too small, the operator would not account for the behavior of HVs so that certain assigned trips may not be taken by the HVs, and thus worsen the performance of the proposed MPC.

\section{Conclusion}
\label{sec:conclusion}
We proposed a Stackelberg game-based framework to investigate 
an MoD system with a mixed fleet of AVs and HVs where AVs can be fully coordinated by the operator, but HVs act on their own interest. 
We  developed two formulations: 1) a steady-state formulation to model the behaviors of HVs and analyze the properties of the problem; and 2) a time-varying formulation to devise a real-time MPC based control algorithm. We further conducted real-world case studies based on a Singapore dataset to validate the proposed formulation and algorithms. Results show that the system-level earnings can be significantly improved (by up to 32\%) by introducing AVs into the MoD system. We further show that by considering the interactions between AVs and HVs and by optimizing the compensations, the system-level earnings can be improved by 12\% and 8\%, respectively. Overall, this sheds light on the promising value of the proposed approach.

This research opens the field for several research directions.
First, we would like to extend the proposed approaches to facilitate the integration with real-world MoD services by considering formulations with different objectives (e.g., operator's profit), higher spatio-temporal resolution,
real-time pricing strategies (e.g., surge pricing), and HVs with heterogeneous properties. On the practical side, the strategy proposed in this work has the potential to improve current MoD systems by globally coordinating a small fleet of paid contractor drivers who act ``as AVs".  
Second, it is of interest to account for service externalities, such as congestion effects of vehicle routing, fuel consumption, emissions, etc., which would be important factors to consider as the market share of the MoD service keeps growing. Third, we would like to consider  interactions with other MoD operators and with public transport operators. 
Fourth, we would like to account for the limited driving capability of AVs during the transition period, in the sense that they may not be capable of handling every type of road condition and thus be restricted to operate within a certain subnetwork. Along this line, this work can also be integrated with a co-design framework~\citep[see][for example]{ZardiniLanzettiEtAl2020} to jointly optimize the services of the MoD systems with the allocation of dedicated or smart infrastructure for AVs to provide guidelines to traffic authorities on how to manage and advocate shared AVs in cities.


\section*{Acknowledgement}
The authors would like to thank S.~Hughes and  M.~Yamato from the Toyota Research Institute  and G.~Lin, P.~Kotamarthi, A.~Madelil, J.~Zhou, B.~Adipradana, Z.~K.~Wong, and V.~Torka from Grab Holding Inc. for sharing the data used for the case study and for helpful discussions. We also would like thank M. Zallio for help with the graphics. This research was supported by the
National Science Foundation (NSF) under CAREER Award CMMI1454737 and the Toyota Research Institute (TRI). The first author would like to acknowledge the support of the Swiss National Science Foundation (SNSF) Postdoc.Mobility Fellowship (P2EZP2\_184244). This article
solely reflects the opinions and conclusions of its authors and not NSF, TRI, SNSF, or any other entity.

 \bibliography{main,ASL_papers}

\appendix
\section{Proofs}

\subsection{Proof for Theorem \ref{thm:HVEquilibrium}}
\label{proof:HVEquilibrium}

\begin{proof}[Proof of Theorem \ref{thm:HVEquilibrium}] We ignore index $h$ for the simplicity of presentation. 

(i)  Integrating Eq.(\ref{eq:HVOpt0}) into Eq.(\ref{eq:HVnumber}), we have
\begin{align}
g(v) &= \sum_{(i,j)\in\mathcal{E}} (x^*_{ij}+y^*_{ij})\tau_{ij}  + \sum_{(i,j)\in\mathcal{E}}u^*_{ij} \notag \\
&= \frac{1}{v}\Big(\sum_{(i,j)\in\mathcal{E}}c_{ij}x^*_{ij} - J^{\rm{SS}*}_{L,v} +  \sum_{(i,j)\in\mathcal{E}}\pi^*_{ij}\bar{q}_{ij}- \sigma\sum_{(i,j)\in \mathcal{E}}\delta_{ij}(x_{ij}^*+y_{ij}^*)\Big)\notag\\
&=\frac{1}{v}\Big(\sum_{(i,j)\in\mathcal{E}} c_{ij}x^*_{ij}- \sigma\sum_{(i,j)\in \mathcal{E}}\delta_{ij}(x_{ij}^*+y_{ij}^*)\Big). \label{eq:HVnumber1}
\end{align}
The last equality is due to the strong duality of Problem~\ref{prb:HV1}, i.e.,  $J^{\rm{SS}*}_{L,v} = \sum_{(i,j)\in\mathcal{E}}\pi^*_{ij}\bar{q}_{ij}$,  where the RHS represents the dual optimum.

 Denote the optimal solutions to Problem~\ref{prb:HV1} with parameters $\bar{v}$ and $\hat{v}$, respectively, as $(\bm{\bar{x}}^*,\bm{\bar{y}}^*)$ and $(\bm{\hat{x}}^*,\bm{\hat{y}}^*)$. Then by the optimality of these solutions, we have, 
\begin{align}
&\sum_{(i,j)\in\mathcal{E}}{c_{ij}\bar{x}^*_{ij}} -  \hat{v}\sum_{(i,j)\in\mathcal{E}} \tau_{ij}(\bar{x}^*_{ij} + \bar{y}^*_{ij}) - \sigma\sum_{(i,j)\in \mathcal{E}}\delta_{ij}(\bar{x}_{ij}^*+\bar{y}_{ij}^*) \notag \\
\leq &  \sum_{(i,j)\in\mathcal{E}}{c_{ij}\hat{x}^*_{ij}} -  \hat{v}\sum_{(i,j)\in\mathcal{E}} \tau_{ij}(\hat{x}^*_{ij} + \hat{y}^*_{ij})- \sigma\sum_{(i,j)\in \mathcal{E}}\delta_{ij}(\hat{x}_{ij}^*+\hat{y}_{ij}^*)\label{eq:jp1}\\
 &\sum_{(i,j)\in\mathcal{E}}{c_{ij}\hat{x}^*_{ij}} -  \bar{v}\sum_{(i,j)\in\mathcal{E}} \tau_{ij}(\hat{x}^*_{ij} + \hat{y}^*_{ij})  - \sigma\sum_{(i,j)\in \mathcal{E}}\delta_{ij}(\hat{x}_{ij}^*+\hat{y}_{ij}^*) \notag \\
 \leq & \sum_{(i,j)\in\mathcal{E}}{c_{ij}\bar{x}^*_{ij}} -  \bar{v}\sum_{(i,j)\in\mathcal{E}} \tau_{ij}(\bar{x}^*_{ij} + \bar{y}^*_{ij})- \sigma\sum_{(i,j)\in \mathcal{E}}\delta_{ij}(\bar{x}_{ij}^*+\bar{y}_{ij}^*).\label{eq:jp2} 
\end{align}

Multiplying Eq.(\ref{eq:jp1}) by $\bar{v}$ and multiplying Eq.(\ref{eq:jp2}) by $\hat{v}$, and summing up the resulting inequalities, we can obtain
\begin{align}
(\bar{v}-\hat{v})\Big(\sum_{(i,j)\in\mathcal{E}}{c_{ij}\bar{x}^*_{ij}}-\sigma\sum_{(i,j)\in \mathcal{E}}\delta_{ij}(\bar{x}_{ij}^*+\bar{y}_{ij}^*) - \sum_{(i,j)\in\mathcal{E}}{c_{ij}\hat{x}^*_{ij}} +\sigma\sum_{(i,j)\in \mathcal{E}}\delta_{ij}(\hat{x}_{ij}^*+\hat{y}_{ij}^*)\Big) \leq 0. 
\end{align}
Since $\bar{v}< \hat{v}$, we have $\bar{v}g(\bar{v}) \geq \hat{v}g(\hat{v})$, and hence $g(\bar{v}) \geq \hat{v}g(\hat{v})/\bar{v}>g(\hat{v})$. By Theorem 4.29 and Theorem 4.30 in \citet{Rudin1964}, for monotonic function $g(v)$, $g(v^-)$ and $g(v^+)$ exists for every $v\geq w$, and the discontinuities are countable. 

(ii) Notice that we can choose sufficiently large $v=\bar{v}$ such that $c_{ij}-\sigma\delta_{ij}-v\tau_{ij}<0,~\forall(i,j)\in\mathcal{E}$. We can then verify that the optimal solution to Problem~\ref{prb:HV1} would be the trivial solution $x_{ij}=0$ and $y_{ij}=0,~(i,j)\in\mathcal{E}$. This implies that $g(\bar{v})=0$. 

Under Assumption~\ref{asm:trail},  there exists $l = (e_1,e_2,\cdots,e_{|l|}) \in \mathcal{L}$, such that $U_l\geq vT_L$, where $e_r = (i_r,j_r,\kappa_r)$. Let $z_l=\min\{q_{ij}|(i,j,{\rm{pax}}) \in l\}>0$. We next show that the optimal solution $(\bm{x}^*,\bm{y}^*)$ to Problem~\ref{prb:HV1} with parameter $v=0$ satisfy
\begin{align}
    \sum_{(i,j)\in \mathcal{E}}c_{ij}x_{ij}^*- \sigma\sum_{(i,j)\in \mathcal{E}}\delta_{ij}(x_{ij}^*+y_{ij}^*) \geq \mu z_lT_l >0
\end{align}

To this end, we set $\{x_{ij},y_{ij}\}_{(i,j)\in\mathcal{E}}$ as 
\begin{align}
    x_{ij} =\left\{ \begin{aligned}
    &z_l,\quad & {\rm{if~}}(i,j,{\rm{pax}}) \in l\\
     &0,\quad & \rm{otherwise} 
    \end{aligned}\right.\\
    y_{ij} =\left\{ \begin{aligned}
    &z_l,\quad & {\rm{if~}}(i,j,{\rm{reb}}) \in l\\
     &0,\quad & \rm{otherwise} 
    \end{aligned}\right.
\end{align}
Clearly, $\{x_{ij},y_{ij}\}_{(i,j)\in\mathcal{E}}$ is a feasible solution to Problem~\ref{prb:HV1} with parameter $v=0$. By Eq.(\ref{eq:U}), Eq.(\ref{eq:T}), and Assumption~\ref{asm:trail}, we further have 
\begin{align}
    &\sum_{(i,j)\in \mathcal{E}}c_{ij}x_{ij}^h- \sigma\sum_{(i,j)\in \mathcal{E}}\delta_{ij}(x_{ij}^h+y_{ij}^h) = z_l\sum_{r=1}^{|l|}\Big((c_{i_rj_r}-\sigma\delta_{i_rj_r})\mathbb{I}_{\kappa_r=\rm{pax}}-\sigma\delta_{i_rj_r}\mathbb{I}_{\kappa_r=\rm{reb}}\Big) = z_lU_l \geq \mu z_lT_l >0
\end{align}
This implies $g(0)=+\infty$. Hence, by the monotonocity of $g(v)$,  for any $N^h>0$, there exists a $v_0> 0$ such that $g(v_0^+)\leq N^h \leq g(v_0^-)$.

(iii) We first handle the case with $v_0<\mu$ or $g(v^+)=g(v^-)$. The Karush–Kuhn–Tucker (KKT) conditions suggest that $(\bm{x}^*,\bm{y}^*)$ and $(\bm{\pi}^*,\bm{\alpha}^*,\bm{\beta}^*,\bm{\phi}^*)$ are the primal and dual solutions to Problem~\ref{prb:HV1}, respectively, if and only if they satisfy the constraints Eq.(\ref{eq:HVOpt1})--Eq.(\ref{eq:HVOpt3}) and the following conditions. 
\begin{subequations}
\begin{align}
0&=\phi^*_{i} - c_{ij} + v\tau_{ij} + \sigma\delta_{ij} + \pi^*_{ij} - \phi^*_{j} -\alpha^*_{ij},~(i,j)\in\mathcal{E}\label{eq:KKT1-1}\\
0&=\phi^*_{i} + v\tau_{ij} + \sigma\delta_{ij}- \phi^*_{j}-\beta^*_{ij} ,~(i,j)\in\mathcal{E} \label{eq:KKT1-2}\\ 
0&= \pi^*_{ij}(x^*_{ij}-\bar{q}_{ij}) ,~(i,j)\in\mathcal{E} \label{eq:KKT1-3}\\
0&=x^*_{ij}\alpha^*_{ij},~(i,j)\in\mathcal{E} \label{eq:KKT1-4}\\ 
0&=y^*_{ij}\beta^*_{ij},~(i,j)\in\mathcal{E} \label{eq:KKT1-5}\\ 
0&\leq \pi^*_{ij},\alpha^*_{ij},\beta^*_{ij},~(i,j)\in\mathcal{E}.  \label{eq:KKT1-6} 
\end{align}\label{eq:KKT1}
\end{subequations}

We next verify that $(\bm{x}^*,\bm{y}^*,\bm{u}^*)$ satisfy the conditions in Definition~\ref{def:HVEquilibrium}. 
For condition (a), we only need to show $(\bm{x}^*,\bm{y}^*)$ satisfy Eq.(\ref{eq:HVWaiting}) and Eq.(\ref{eq:HVUB}). Recall that $u^*_{ij} = \pi^*_{ij}\bar{q}_{ij}/v\geq 0$, and by Eq.(\ref{eq:KKT1-3}), Eq.(\ref{eq:HVWaiting}) holds.  Since $g(v)$ is strictly decreasing and $v=\{v_0,\mu\}$, we have $g(v)\leq g(v_0^+)\leq N^h$, hence Eq.(\ref{eq:HVUB}) holds.
 
For condition (b), for any used trail $l=(e_1,\cdots,e_{|l|})$, since $\displaystyle z_l=  \min_{1\leq r\leq |l|}\{x_{i_rj_r}\mathbb{I}_{\kappa_r=\rm{pax}} + y_{i_rj_r}\mathbb{I}_{\kappa_r=\rm{reb}}\}> 0$, we have $x_{i_rj_r}>0$ if $\kappa_r=\rm{pax}$ and $y_{i_rj_r}>0$ if $\kappa_r=\rm{reb}$. Then by Eq.(\ref{eq:KKT1-1}), Eq.(\ref{eq:KKT1-2}), Eq.(\ref{eq:KKT1-4}) and Eq.(\ref{eq:KKT1-5}), we have 
\begin{align}
\phi^*_{i_r} - \phi^*_{j_r} = \left\{
\begin{aligned}
&c_{i_rj_r} - \sigma\delta_{i_rj_r}- v\tau_{i_rj_r} - \pi^*_{i_rj_r},&{\rm{if}}~\kappa_r=\rm{pax} \\
&- \sigma\delta_{i_rj_r}-v\tau_{i_rj_r},&{\rm{if}}~\kappa_r=\rm{reb}.   
\end{aligned}\right.\label{eq:phi}
\end{align}
Then summing up Eq.(\ref{eq:phi}) along trail $l$, we have
\begin{align}
\sum_{r=1}^{|l|}{\Big(c_{i_rj_r} - \sigma\delta_{i_rj_r}- v\tau_{i_rj_r} - \pi^*_{i_rj_r}\Big)\mathbb{I}_{\kappa_r=\rm{pax}} + \Big(- \sigma\delta_{i_rj_r}-v\tau_{i_rj_r}\Big)\mathbb{I}_{\kappa_r=\rm{reb}}} = 0,
\end{align}
which implies $U_l-vT_l=0$, hence $U_l/T_l=v$. Since $v=\max\{v_0,\mu\}$, we have $U_l/T_l \geq \mu$. For any trail $l'=(e_1',\cdots,e_{|l|}')\in\mathcal{L}$, by Eq.(\ref{eq:KKT1-1}), Eq.(\ref{eq:KKT1-2}), and Eq.(\ref{eq:KKT1-6}), we can similar obtain  $U_{l'}/T_{l'} \leq v$. Hence, it holds that $U_{l'}/T_{l'}\leq U_l/T_l$. 

For condition (c), the case $v_0<\mu$ or $g(v^+)=g(v^-)$ implies that $v=\mu$ or $g(v)=N^h$ .  

Therefore, $\{x^*_{ij},y^*_{ij},u^*_{ij}\}_{(i,j)\in\mathcal{E}}$ is an equilibrium and $v$ is the earning rate at the equilibrium. 

For the case with $g(v_0^+)\leq N^h \leq g(v_0^-)$, we first show that any solution to Problem~\ref{prb:HV2} with parameter $v_0$ is also an optimal solution to Problem~\ref{prb:HV2} with parameter $v_0$. Clearly, any optimal solution to Problem~\ref{prb:HV2} is a feasible solution to Problem~\ref{prb:HV1}. We show that Problem~\ref{prb:HV1} has the same optimum as Problem~\ref{prb:HV2}.  In the case with $g(v_0^+)<g(v_0^-)$, Problem~\ref{prb:HV1} has multiple solutions. Let $\bm{x}_+^*$ be the optimal solution associated with $v_0^+$, and $\bm{x}_-^*$ be the optimal solution associated with $v_0^-$.  Since Problem~\ref{prb:HV2} is a LP, any convex combination of $\bm{x}_+^*$ and $\bm{x}_-^*$  is also a solution to Problem~\ref{prb:HV2}. Also notice that for given $v$, $g(v)$ is a continuous function of $\bm{x}^*$. Hence, there exists an optimal solution $\bm{x}^*$ to Problem~\ref{prb:HV2} satisfying 
\begin{align}
g(v) =\frac{1}{v}\Big(\sum_{(i,j)\in\mathcal{E}} c_{ij}x^*_{ij}- \sigma\sum_{(i,j)\in \mathcal{E}}\delta_{ij}(x_{ij}^*+y_{ij}^*)\Big) =N^h. \label{eq:HVnumber2}
\end{align}
That being said, $\bm{x}^*$ is  also a solution to Problem~\ref{prb:HV2}. Hence, Problem~\ref{prb:HV1} and Problem~\ref{prb:HV2} have the same optimum, i.e., $\tilde{J}^{\rm{SS}*}_{L,v}= J^{\rm{SS}*}_{L,v}$. Therefore, any optimal solution to Problem~\ref{prb:HV2} $ (\tilde{\bm{x}}^*,\tilde{\bm{y}}^*)$ is also an optimal solution to Problem~\ref{prb:HV1}.  

Then, for the optimal solution $ (\tilde{\bm{x}}^*,\tilde{\bm{y}}^*)$ in Problem~\ref{prb:HV1}, let $\bm{\tilde{\pi}}^*$ be the optimal dual variable corresponding to constraint Eq.(\ref{eq:HVOpt1}). Set the number of waiting HVs $\bm{\tilde{u}}^*$ as $\tilde{u}_{ij}=\tilde{\pi}_{ij}\bar{q}_{ij}/v$. Similar to the proof for the other case, we can derive the KKT conditions for Problem~\ref{prb:HV1} for solution $ (\tilde{x}^*,\tilde{y}^*)$ and show that  $ (\bm{\tilde{x}}^*,\bm{\tilde{y}}^*,\bm{\tilde{u}}^*)$ satisfies the conditions in Definition~\ref{def:HVEquilibrium}. 

This concludes the proof. 
 \end{proof}

\subsection{Proof for Theorem~\ref{thm:mixed}}
\label{proof:mixed}
\begin{proof}[Proof for Theorem~\ref{thm:mixed}]
Without the loss of generality, we assume $H(\cdot)$ is differentiable within range $[0,\eta_{ij}]$. Otherwise, we can replace the gradients with subgradients in the proof. 
Since Problem~\ref{prb:global} is convex,  the KTT conditions suggest that $(\bm{x}^{a*},\bm{y}^{a*},\bm{x}^{h*},\bm{y}^{h*},\bm{q}^*)$ and $(\bm{\phi}^*,\bm{\pi}^*, \bm{\gamma}^*,\bm{\alpha}^*,\bm{\beta}^*,\bm{\theta}^*,\bm{\rho}^*)$ are the optimal primal and dual solutions to Eq.(\ref{eq:globalOpt}), respectively, if and only if they satisfy the primal constraints Eq.(\ref{eq:AVDemand2})--Eq.(\ref{eq:AVBounds2}) and the following conditions.  
\begin{subequations}
\begin{align}
0&=- H_{ij}(q^*_{ij}) - q^*_{ij}H'_{ij}(q^*_{ij}) + \pi^{*}_{ij} + \theta^{*}_{ij} - \rho^{*}_{ij},~i\neq j\in\mathcal{N}  \label{eq:KKT2-1} \\
0&= \tau_{ij}\mu\mathbb{I}_{m=h} + \sigma\delta_{ij}  + \phi^{m*}_i - \phi^{m*}_j  + \tau_{ij}\gamma^{m*} - \pi^{*}_{ij}- \alpha_{ij}^{m*},~i\neq j\in\mathcal{N},~m\in\mathcal{M} \label{eq:KKT2-2} \\
0&=\tau_{ij}\mu\mathbb{I}_{m=h} + \sigma\delta_{ij}    + \phi^{m*}_i - \phi^{m*}_j + \tau_{ij}\gamma^{m*} - \beta_{ij}^{m*},~i\neq j\in\mathcal{N},~m\in\mathcal{M} \label{eq:KKT2-3} \\
0&=q^*_{ij}\rho^{*}_{ij},~i\neq j\in\mathcal{N} \label{eq:KKT2-4} \\
0&=(q^*_{ij} - \eta_{ij})\theta^{*}_{ij},~i\neq j\in\mathcal{N} \label{eq:KKT2-5}\\
0&=x^{m*}_{ij}\alpha_{ij}^{m*},~i\neq j\in\mathcal{N},~m\in\mathcal{M} \label{eq:KKT2-6} \\
0&=y^{m*}_{ij}\beta_{ij}^{m*},~i\neq j\in\mathcal{N},~m\in\mathcal{M} \label{eq:KKT2-7} \\
0&\leq\bm{\gamma}^*,\bm{\alpha}^*,\bm{\beta}^*,\bm{\theta}^*,\bm{\rho}^*.  \label{eq:KKT2-8}
\end{align}
\end{subequations}

Let $c^*_{ij}$ satisfy
\begin{align}
c^*_{ij} = \max\{\pi^{*}_{ij},0\} \label{eq:compensation} . 
\end{align} 

If $0<q^*_{ij}\leq \eta_{ij}$, notice that $H'_{ij}(q^*_{ij})<0$ and $H_{ij}(q^*_{ij})\geq 0$, then by Eq.(\ref{eq:KKT2-1}), Eq.(\ref{eq:KKT2-4}), Eq.(\ref{eq:KKT2-8}), and Eq.(\ref{eq:compensation}), we have
\begin{align}
0 \leq c^*_{ij} = \max\{\pi^{*}_{ij},0\} = \max\{H_{ij}(q^*_{ij})  + q^*_{ij}H'_{ij}(q^*_{ij}) - \theta_{ij},0\} \leq H_{ij}(q^*_{ij}) = p^*_{ij}
\end{align}

We first show that $(\bm{x}^{h*},\bm{y}^{h*},\bm{u})$ is an equilibrium for HVs under the vector of compensations and remaining demand $(\bm{c},\bm{\bar{q}})$, giving an earning rate of $\mu+\gamma\geq \mu$. To this end, we verify the conditions in Definition \ref{def:HVEquilibrium}. For condition (a) in Definition~\ref{def:HVEquilibrium}, notice that $u^{h*}_{ij}=0,~i\neq j\in\mathcal{N}$, then the the waiting time $\zeta_{ij}=0$, and by constraints Eq.(\ref{eq:AVFlow2}), Eq.(\ref{eq:AVDemand2}), and Eq.(\ref{eq:AVSize2}), we can derive Eq.(\ref{eq:HVDemand})-Eq.(\ref{eq:HVUB}). For condition (b), notice that $\alpha_{ij}^{h*}+\pi^{*}_{ij}=\beta_{ij}^{h*}$ by comparing Eq.(\ref{eq:KKT2-2}) and Eq.(\ref{eq:KKT2-3}), we have for any trail $l'$,
\begin{align}
U_{l'} - \mu T_{l'} &=  
\sum_{r=1}^{|l'|}{\Big((c^*_{i_rj_r} - (\mu+\gamma)\tau_{i_rj_r}-\sigma\delta_{i_rj_r})\mathbb{I}_{\kappa_r=\rm{pax}} + (-(\mu+\gamma)\tau_{i_rj_r}-\sigma\delta_{i_rj_r})\mathbb{I}_{\kappa_r=\rm{reb}}\Big)} \notag \\
&= 
\sum_{r=1}^{|l'|}{\Big((\max\{\pi^{*}_{i_rj_r},0\} + \phi^{h*}_{i_r} - \phi^{h*}_{j_r} - \pi^{*}_{i_rj_r} -\alpha^{h*}_{i_rj_r} )\mathbb{I}_{\kappa_r=\rm{pax}} +(\phi^{m*}_{i_r} - \phi^{m*}_{j_r}-\beta^{h*}_{i_rj_r})\mathbb{I}_{\kappa_r=\rm{reb}} \Big)} \notag \\
& = \sum_{r=1}^{|l'|}{\Big(\max\{-\alpha^{h*}_{i_rj_r},-\beta^{h*}_{i_rj_r}\}\mathbb{I}_{\kappa_r=\rm{pax}} -\beta^{h*}_{i_rj_r}\mathbb{I}_{\kappa_r=\rm{reb}} \Big)}
\end{align}
By Eq.(\ref{eq:KKT2-8}), we have $U_{l'} - \mu T_{l'} \leq 0$ . If trail $l$ is used, since $\alpha^{h*}_{i_rj_r}=\beta^{h*}_{i_rj_r}=0,\forall r$, we have $U_{l} - \mu T_{l} = 0$. Hence, condition (c) also holds. 

We next prove that $(\bm{q}^*,\bm{c}^*,\bm{x}^{a*},\bm{y}^{a*})$ is an optimal solution to Problem~\ref{prb:operatorOpt}. 
Clearly, is a feasible solution. Hence, $\tilde{J}^{\rm{SS}}_L(\bm{q}^*,\bm{c}^*,\bm{x}^{a*},\bm{y}^{a*}) \leq J^{\rm{SS},*}_L$. 
For any solution $(\bm{q},\bm{c},\bm{x}^{a},\bm{y}^{a})$ satisfying Eq.(\ref{eq:AVDemand2}) -- Eq.(\ref{eq:AVBounds2}), let $(\bm{x}^{h},\bm{y}^{h},\bm{u}^{h})$ be the induced HV equilibrium and $v$ is the earning rate at the equilibrium. 

We first construct the following problem by eliminating $\bm{u}$ from the objective function. 
\begin{subequations}
 \begin{align}
\max_{\bm{q},\bm{x}^a,\bm{y}^a,\bm{x}^h,\bm{y}^h}&\quad  P_1 = \quad  \sum_{(i,j)\in \mathcal{E}}H_{ij}(q_{ij})(x^a_{ij}+x^h_{ij}) -\mu\sum_{(i,j)\in\mathcal{E}} (x_{ij}^h+y_{ij}^h)\tau_{ij} 
- 
\sigma\sum_{m\in\mathcal{M}}\sum_{(i,j)\in \mathcal{E}}\delta_{ij}(x^m_{ij}+y^m_{ij}) \label{eq:AVOpt1obj}\\
\rm{s.t.}& \quad \rm{Eq.(\ref{eq:HVCouple}),~Eq.(\ref{eq:AVFlow2}),~Eq.(\ref{eq:AVSize2}),~Eq.(\ref{eq:AVBounds2})}\notag \\
&\quad   x_{ij}^a + x_{ij}^h \leq q_{ij},(i,j)\in\mathcal{E}  
\end{align}\label{eq:AVOpt1} 
\end{subequations}
Since the constraints for HVs are already satisfied by the follower model Problem~\ref{prb:HV1} or Problem~\ref{prb:HV2}, we can equivalently add these constraints to the optimization problem. Also notice that for the same feasible solution, the objective $P_1$ is no less than $J^{\rm{SS}}_L$.  Thus, we have $J^{\rm{SS},*}_L \leq P^*_1$.

 Problem Eq.(\ref{eq:AVOpt2}) removes constraints Eq.(\ref{eq:HVCouple}) from Problem Eq.(\ref{eq:AVOpt1}), hence we have $P^*_1\leq P^*_2$. 
  \begin{align}
\max_{}&\quad  P_2 = \rm{Eq.(\ref{eq:AVOpt1obj})} \notag \\
\rm{s.t.}& \quad \rm{Eq.(\ref{eq:AVFlow2}),~Eq.(\ref{eq:AVSize2}),~Eq.(\ref{eq:AVBounds2})} \notag \\
&\quad   x_{ij}^a + x_{ij}^h \leq q_{ij},(i,j)\in\mathcal{E} \label{eq:AVOpt2}
\end{align}

 Problem Eq.(\ref{eq:AVOpt3}) modifies the inequality constraints for demand in  Problem Eq.(\ref{eq:AVOpt2}) to equality constraints Eq.(\ref{eq:AVDemand2}). 
   \begin{align}
\max&\quad  P_3 = \rm{Eq}.(\ref{eq:AVOpt1obj}) \notag \\
\rm{s.t.}& \quad \rm{Eq.(\ref{eq:AVDemand2}) - Eq.(\ref{eq:AVBounds2})} \label{eq:AVOpt3}
\end{align}
We notice that for given $\bm{x},\bm{y}$, $P_2$ increases as $\bm{q}$ is reduced. Hence, $P_2^*\leq P_3^*$. Further notice that  Problem Eq.(\ref{eq:AVOpt3}) is equivalent to Problem~\ref{prb:global} if Eq.(\ref{eq:AVDemand2}) hold, and thus hence $P_3^*= \tilde{J}^{\rm{SS}}_L(\bm{q}^*,\bm{c}^*,\bm{x}^{a*},\bm{y}^{a*})$. 
Therefore, $\tilde{J}^{\rm{SS}}_L(\bm{q}^*,\bm{c}^*,\bm{x}^{a*},\bm{y}^{a*})=P^*_1\leq P^*_2\leq P_3^*= \tilde{J}^{\rm{SS}}_L(\bm{q}^*,\bm{c}^*,\bm{x}^{a*},\bm{y}^{a*})$. This completes the proof. 

 \end{proof}
 
\subsection{Proof of Proposition~\ref{prp:implication}}\label{proof:implication}
\begin{proof}[Proof of Proposition~\ref{prp:implication}]
Let $(\bm{x}^*,\bm{y}^*)$ be a solution to Problem~\ref{prb:HV1}. We only need to show that $\bm{x}^*$ is also an optimal solution to Problem~\ref{prb:rvfollower}. The Karush–Kuhn–Tucker (KKT) conditions suggest that $\bm{x}^*$ is  the primal solution to Problem~\ref{prb:rvfollower}, if and only if there exists a vector $(\bm{\pi}^*,\bm{\alpha}^*,\bm{\chi}^*)$ such that  Eq.(\ref{eq:rvfollower1})--Eq.(\ref{eq:rvfollower2}) and the following conditions hold. 
\begin{subequations}
\begin{align}
0&=\phi^*_{i} - c_{ij} + v\tau_{ij}+\sigma\delta_{ij} + \pi^*_{ij} - \phi^*_{j} -\alpha^*_{ij} + \beta^*_{ij},~(i,j)\in\mathcal{E}\label{eq:KKT3-1}\\
0&= \pi^*_{ij}(x^*_{ij}-\bar{q}_{ij}) ,~(i,j)\in\mathcal{E} \label{eq:KKT3-2}\\
0&=x^*_{ij}\alpha^*_{ij},~(i,j)\in\mathcal{E} \label{eq:KKT3-3}\\ 
0&\leq \pi^*_{ij},\alpha^*_{ij},~(i,j)\in\mathcal{E}.  \label{eq:KKT3-4} 
\end{align}\label{eq:KKT3}
\end{subequations}
Comparing Eq.(\ref{eq:KKT3}) with Eq.(\ref{eq:KKT1}), we can set $\bm{\pi}^*$, $\bm{\alpha}^*$, and $\bm{\phi}^*$ as the optimal dual solutions to Problem~\ref{prb:HV1}, and $\bm{\chi}=0$, and then conditions Eq.(\ref{eq:KKT1}) hold. Therefore,  $\bm{x}^*$  is an optimal solution to Problem~\ref{prb:rvfollower}. 
\end{proof}

\end{document}